\documentclass[11pt]{article}

\usepackage[margin=1in]{geometry}

\DeclareMathAlphabet{\mathbbold}{U}{bbold}{m}{n}
\usepackage{amsthm,amsmath,amsfonts,amssymb}
\usepackage[usenames,dvipsnames,svgnames,table]{xcolor}
\usepackage{thm-restate}

\usepackage[pagebackref]{hyperref}
\hypersetup{
    bookmarksnumbered=true, % If Acrobat bookmarks are requested, include section numbers
    unicode=false, % non-Latin characters in Acrobat’s bookmarks
    pdfstartview={FitH}, % fits the width of the page to the window
    pdftitle={Title}, % title
    pdfauthor={Authors}, % author
    pdfsubject={}, % subject of the document
    pdfcreator={}, % creator of the document
    pdfproducer={}, % producer of the document
    pdfkeywords={}, % list of keywords
    pdfnewwindow=true, % links in new window
    colorlinks=true, % false: boxed links; true: colored links
    linkcolor=blue, % color of internal links
    citecolor=blue, % color of links to bibliography
    filecolor=blue, % color of file links
    urlcolor=blue % color of external links
}

\renewcommand{\backref}[1]{}

\renewcommand{\backrefalt}[4]{%
\ifcase #1 %
\or
[p.\ #2]%
\else
[pp.\ #2]%
\fi}

\usepackage{tikz,tikz-qtree}
\usepackage{colortbl}
\usepackage{multirow}
\usepackage{array}
\usepackage{wrapfig}
\usepackage{enumitem}

\usepackage{tocloft}

\newtheorem{theorem}{Theorem}
\newtheorem{lemma}[theorem]{Lemma}

\newtheorem{definition}[theorem]{Definition}

\newcommand{\eq}[1]{\hyperref[eq:#1]{(\ref*{eq:#1})}}
\renewcommand{\sec}[1]{\hyperref[sec:#1]{Section~\ref*{sec:#1}}}
\newcommand{\thm}[1]{\hyperref[thm:#1]{Theorem~\ref*{thm:#1}}}
\newcommand{\lem}[1]{\hyperref[lem:#1]{Lemma~\ref*{lem:#1}}}
\newcommand{\prop}[1]{\hyperref[prop:#1]{Proposition~\ref*{prop:#1}}}
\newcommand{\cor}[1]{\hyperref[cor:#1]{Corollary~\ref*{cor:#1}}}
\newcommand{\fig}[1]{\hyperref[fig:#1]{Figure~\ref*{fig:#1}}}
\newcommand{\tab}[1]{\hyperref[tab:#1]{Table~\ref*{tab:#1}}}
\newcommand{\alg}[1]{\hyperref[alg:#1]{Algorithm~\ref*{alg:#1}}}
\newcommand{\app}[1]{\hyperref[app:#1]{Appendix~\ref*{app:#1}}}

\newcommand{\B}{\{0,1\}}

\newcommand{\BKK}{\textsc{BKK}}
\newcommand{\RBKK}{\textsc{RecBKK}}
\newcommand{\ksum}{k\textsc{-sum}}
\newcommand{\ANDOR}{\textsc{And-Or}}
\newcommand{\AND}{\textsc{And}}
\newcommand{\NAND}{\textsc{Nand}}
\newcommand{\OR}{\textsc{Or}}
\newcommand{\PARITY}{\textsc{Parity}}

\renewcommand{\th}[1]{$#1^\mathrm{th}$}

\DeclareMathOperator{\adeg}{\widetilde{\deg}}
\DeclareMathOperator{\bs}{bs}
\DeclareMathOperator{\RC}{RC}

\DeclareMathOperator{\Adv}{Adv^{\pm}}
\DeclareMathOperator{\crit}{crit}

\DeclareMathOperator{\poly}{poly}
\DeclareMathOperator{\Dom}{Dom}

\newcommand{\tO}{\widetilde{O}}
\newcommand{\tOmega}{\widetilde{\Omega}}
\newcommand{\tTheta}{\widetilde{\Theta}}
\newcommand{\CS}{\mathrm{CS}}
\newcommand{\CSp}{\mathrm{CS(\phi)}}

\definecolor{conj}{HTML}{C2C0BF}
\definecolor{open}{HTML}{A31F34}

\newcommand{\ct}[2]{%
\begin{tabular}[t]{@{}l@{}}#1\\#2\end{tabular}}
\newcommand{\co}[2]{%
\cellcolor{open!70}\begin{tabular}[t]{@{}l@{}}#1\\#2\end{tabular}}
\newcommand{\cc}[2]{%
\cellcolor{conj!70}\begin{tabular}[t]{@{}l@{}}#1\\#2\end{tabular}}
\newcommand{\smcite}[1]{{\small\cite{#1}}}

\begin{document}
%%%%%%%%%%%%%%%%%%%%%%%%%%%%%%%%%%%%%%%%%%%%%%%%%%%%%%%%%%%%%%%%%%%%%%%%%%%%%%

\title{\vspace{-1em} \bfseries Separations in query complexity using cheat sheets}

\author{
Scott Aaronson\\
\small MIT\\
\small \texttt{aaronson@csail.mit.edu}
\and
Shalev Ben-David\\
\small MIT\\
\small \texttt{shalev@mit.edu}
\and
Robin Kothari \\
\small MIT\\
\small \texttt{rkothari@mit.edu}
}

\date{}
\maketitle

%%%%%%%%%%%%%%%%%%%%%%%%%%%%%%%%%%%%%%%%%%%%%%%%%%%%%%%%%%%%%%%%%%%%%%%%%%%%%%
\begin{abstract}
We show a power $2.5$ separation between bounded-error randomized and quantum query complexity for a total Boolean function, refuting the widely believed conjecture that the best such separation could only be quadratic (from Grover's algorithm). \
We also present a total function with a power $4$ separation between quantum query complexity and approximate polynomial degree, showing severe limitations on the power of the polynomial method. \
Finally, we exhibit a total function with a quadratic gap between quantum query complexity and certificate complexity, which is optimal (up to log factors). \
These separations are shown using a new, general technique that we call the cheat sheet technique. \ The technique is based on a generic transformation that converts any (possibly partial) function into a new total function with desirable properties for showing separations. \ The framework also allows many known separations, including some recent breakthrough results of Ambainis et al., to be shown in a unified manner.
\end{abstract}

\setlength{\cftbeforesecskip}{0.2ex}
\setlength{\cftaftertoctitleskip}{2ex}
\setlength{\cftbeforetoctitleskip}{1.5ex}
\tableofcontents

\clearpage
%%%%%%%%%%%%%%%%%%%%%%%%%%%%%%%%%%%%%%%%%%%%%%%%%%%%%%%%%%%%%%%%%%%%%%%%%%%%%%
\section{Introduction}
\label{sec:intro}

Query complexity (or decision tree complexity) is a model of computation that allows us to examine the strengths and weaknesses of resources such as access to randomness, nondeterminism, quantum computation and more. \
As opposed to the Turing machine model where we can only conjecture that certain resources exponentially speed up computation, in this model such beliefs can be proved. \
In particular, the query model is a natural setting to describe several well-known quantum algorithms, such as Grover's algorithm \cite{Gro96} and quantum walks~\cite{Amb07}; even Shor's algorithm for integer factorization \cite{Sho97} is based on a quantum query algorithm.

In the query model we measure a problem's complexity, called its query complexity, by the minimum number of input bits that have to be read to solve the problem on any input. \
For a function $f$, we use $D(f)$, $R(f)$, and $Q(f)$ to denote the query complexity of computing $f$ with a deterministic, bounded-error randomized, and bounded-error quantum algorithm respectively.

For a total Boolean function $f:\B^n \to \B$, we know that these measures are polynomially related:
in one direction, we have $Q(f) \leq R(f) \leq D(f)$
as each model is more powerful than the next,
and in the other direction, we have
$D(f) = O(R(f)^3)$ \cite{Nis91} and $D(f) = O(Q(f))^6$ \cite{BBC+01}.

Until recently, the best known separation between $D(f)$ and $R(f)$ was the \NAND-tree function \cite{SW86}, which satisfies $D(f) = \Omega(R(f)^{1.3267})$ and even $D(f) = \Omega(R_0(f)^{1.3267})$, where $R_0(f)$ denotes zero-error randomized query complexity. \
The best known separation between $D(f)$ or $R(f)$ versus $Q(f)$ was only quadratic, achieved by the $\OR$ function: $D(\OR) \geq R(\OR)= \Omega(n)$ and $Q(\OR) = \Theta(\sqrt{n})$ \cite{Gro96,BBBV97}. \ No separation was known between $R_0(f)$ and $R(f)$. \ Furthermore, it was believed that all of these separations were optimal.

These beliefs were shattered by a recent breakthrough by Ambainis, Balodis, Belovs, Lee, Santha, and Smotrovs~\cite{ABB+15}, who built upon the  techniques of G{\"o}{\"o}s, Pitassi, and Watson~\cite{GPW15}, and
 showed the near-optimal separations $D(f) = \tOmega(R_0(f)^2)$ and $R_0(f) = \tOmega(R(f)^2)$. \
Additionally, they presented new separations between quantum and classical models, exhibiting functions with $D(f) = \tOmega(Q(f)^4)$ and $R_0(f) = \tOmega(Q(f)^3)$. \ However, no super-quadratic separation was shown between a quantum and classical model when they are both allowed to err in the same way (e.g., bounded error, zero error, etc.) \ For example, no super-quadratic separation was shown between $Q(f)$ and its classical analogue, $R(f)$.

%%%%%%%%%%%%%%%%%%%%%%%%%%%%%%%%%%%%%%%%%%%%%%%%%%%%%%%%%%%%%%%%%%%%%%%%%%%%%%
\subsection{Results}

\subsubsection*{Randomized versus quantum query complexity}

We present the first super-quadratic separation between quantum and classical query complexity when both models are allowed bounded error. \
This is a counterexample to a conjecture that had been widely believed about the nature of quantum speed-ups: that if the function is total, so that the inputs are not handpicked to orchestrate a quantum speed-up, then the best quantum advantage possible is the quadratic speed-up achieved by Grover's algorithm.

\begin{restatable}{theorem}{RvsQ}
\label{thm:RvsQ}
There exists a total function $f$ such that
$R(f) = \tOmega(Q(f)^{2.5})$.
\end{restatable}

\thm{RvsQ} can be boosted  to a cubic separation, i.e., $R(f) = \tOmega(Q(f)^{3})$, in a completely black-box manner if there exists a partial function (a function defined only on a subset of $\B^n$) with $\tOmega(n)$ randomized query complexity but only $\poly(\log(n))$ quantum query complexity (the best possible separation between these measures up to log factors). \ It is conjectured that a recently studied partial function called $k$-fold Forrelation achieves this separation \cite{AA15}.

\subsubsection*{Quantum query complexity versus polynomial degree}

Nisan and Szegedy \cite{NS95} introduced two new measures of a Boolean function $f$ called the degree and approximate degree, denoted $\deg(f)$ and $\adeg(f)$ respectively. \
The (approximate) degree of a function $f$ is the minimum degree of a polynomial over the input variables that (approximately) equals $f(x)$ at every input $x\in\B^n$. \
They introduced these measures to lower bound classical query complexity, showing that $\deg(f) \leq D(f)$ and $\adeg(f) \leq R(f)$. \
It turns out that these measures also lower bound the corresponding quantum measures, since $\deg(f) \leq 2\, Q_E(f)$ and $\adeg(f) \leq 2\, Q(f)$ \cite{BBC+01}, where $Q_E(f)$ denotes the exact quantum query complexity of $f$.

Approximate degree has proved to be a fruitful lower bound technique for quantum query complexity, especially for problems like the collision and element distinctness problems \cite{AS04}, where it is known that the original quantum adversary method of Ambainis, another commonly used lower bound technique, cannot show nontrivial lower bounds.

For any lower bound technique, it is natural to ask whether there exist functions where the technique fails to prove a tight lower bound. \
Answering this question, Ambainis showed that the approximate degree of a function can be asymptotically smaller than its quantum query complexity by exhibiting a function with $Q(f) = \Omega(\adeg(f)^{1.3219})$ \cite{Amb03}. \
We dramatically strengthen this separation to obtain nearly a \th{4} power gap.

\begin{restatable}{theorem}{Qvsadeg}
\label{thm:Qvsadeg}
There exists a total function $f$ such that $Q(f) \geq \adeg(f)^{4-o(1)}$.
\end{restatable}

\thm{Qvsadeg} is optimal assuming the conjecture that $D(f) = O(\bs(f)^{2})$, where $\bs(f)$ is block sensitivity (to be defined in \sec{MEASURES}).

\subsubsection*{The cheat sheet technique}

These separations are shown using a new technique for proving
separations between query measures, which we call the
cheat sheet technique. \
Our technique is based on a generic transformation that takes any (partial or total) function and transforms it into a ``cheat sheet version'' of the function that has desirable properties for proving separations.

While the strategy is inspired by the recent breakthrough results \cite{GPW15,ABB+15} (and bears some similarity to older works \cite{LeG06,AdW14}), it represents a more general approach:
it provides a framework for proving separations and allows many separations to be shown in a unified manner. \
Thus the task of proving separations is reduced to the simpler task of finding the right function to plug into this construction. \
For example, it can be used to convert a partial function separation between two models into a weaker total function separation.

In this paper we demonstrate the power of the cheat sheet technique by using it to exhibit several new total function separations in query complexity.

\subsubsection*{Other separations}

On the path to proving \thm{Qvsadeg}, we show several new separations. \ First, we quadratically separate  quantum query complexity from  certificate complexity, which is essentially optimal since $Q(f) \leq  C(f)^2$.

\begin{restatable}{theorem}{QvsC}
\label{thm:QvsC}
There exists a total function $f$ such that $Q(f) = \tOmega(C(f)^{2})$.
\end{restatable}

Besides being a stepping stone to proving \thm{Qvsadeg}, the question of $Q(f)$ versus $C(f)$ has been studied because it is known that the original adversary method of Ambainis (also known as the positive-weights adversary method) cannot prove a lower bound greater than $C(f)$ \cite{SS06,HLS07}.

Plugging the function of \thm{QvsC} into the cheat sheet framework directly yields a function whose quantum query complexity is quadratically larger than its (exact) degree, improving the recent result of \cite{GPW15}, who exhibited a function with $D(f) = \tOmega(\deg(f)^2)$.

\begin{restatable}{theorem}{Qvsdeg}
\label{thm:Qvsdeg}
There exists a total function $f$ such that $Q(f) = \tOmega(\deg(f)^{2})$.
\end{restatable}

This theorem works in a very black-box way: any separation between a measure like $Q(f)$ or $R(f)$ and $C(f)$ can be plugged into the cheat sheet technique to obtain the same separation (up to log factors) between that measure and exact degree. \
For example, since the $\ANDOR$ function satisfies $R(f) = \Omega (C(f)^2)$, the cheat sheet version of $\ANDOR$ satisfies $R(f) = \tOmega(\deg(f)^2)$.

This result also shows limitations on using $\deg(f)$ to lower bound  $Q_E(f)$, since $\deg(f)$ can sometimes be quadratically smaller than $Q_E(f)$ and can even be quadratically smaller than $Q(f)$.

\subsubsection*{Summary of results}

We summarize our new results in the table below.

\setlength{\tabcolsep}{15pt}
\renewcommand{\arraystretch}{1.5}

 \begin{table}[ht]
 \centering
 \begin{tabular}{l l l}

 Separation achieved & Known relation & Result \\ \hline \hline
 $R(f) = \tOmega(Q(f)^{2.5})$ & $R(f) = O(Q(f)^6)$ & \thm{RvsQ} \\ \hline
 $Q(f) \geq \adeg(f)^{4-o(1)}$ & $Q(f) = O(\adeg(f)^{6})$ & \thm{Qvsadeg} \\ \hline
 $Q(f) = \tOmega(C(f)^{2})$ & $Q(f) = O(C(f)^2)$ & \thm{QvsC} \\ \hline
 $Q(f) = \tOmega(\deg(f)^{2})$ & $Q(f) = O(\deg(f)^{3})$ & \thm{Qvsdeg} \\ \hline
\end{tabular}
\caption{New separations shown in this paper}
\label{tab:newsep}
 \end{table}

We are also able to use the cheat sheet technique to reprove many of the query separations of \cite{GPW15,ABB+15,GJPW15}. \ 
Some of their results are subsumed by the results in \tab{newsep}. \ 
We also prove $R_0(f) = \tOmega(Q(f)^{3})$ (\thm{R0vsQ}) and $R(f) = \tOmega(Q_E(f)^{3/2})$ (\thm{RvsQE}) in \sec{oldsep}. \ 
For more details, see \tab{sep} in which we summarize our results and the best known separations and relations between query measures. \ 

It is also worth pointing out what we are unable to reproduce with our framework. The only separations from these papers that we do not reproduce are those
that make essential use of ``back pointers''. \
Reproducing these separations in the cheat sheet framework is an interesting direction for future research. \ 

\setlength{\tabcolsep}{2.8pt}
\renewcommand{\arraystretch}{1.25}

\begin{table}
\begin{minipage}{\linewidth}
\begin{center}
\caption{Best known separations between complexity measures}
\vspace{1em}
\begin{tabular}{l||l|l|l|l|l|l|l|l|l|l}

{} & $D$ & $R_0$ & $R$ & $C$ & $\RC$ & $\bs$ & $Q_E$ & $\deg$ & $Q$ & $\adeg$ \\ \hline\hline

\raisebox{-\height}{$D$} & \cellcolor{darkgray} & \ct{2, 2}{\smcite{ABB+15}} & \ct{2*, 3}{\smcite{ABB+15}} &
\ct{2, 2}{$\wedge\circ\vee$} & \ct{2*, 3}{$\wedge\circ\vee$} & \ct{2*, 3}{$\wedge\circ\vee$} &
\ct{2, 3}{\smcite{ABB+15}} & \cc{2, 3}{\smcite{GPW15}} & \ct{4*, 6}{\smcite{ABB+15}} &
\cc{4*, 6}{\smcite{ABB+15}} \\ \hline

\raisebox{-\height}{$R_0$} & \ct{1, 1}{$\oplus$} &\cellcolor{darkgray} & \ct{2, 2}{\smcite{ABB+15}} &
\ct{2, 2}{$\wedge\circ\vee$} & \ct{2*, 3}{$\wedge\circ\vee$} & \ct{2*, 3}{$\wedge\circ\vee$} &
\ct{2, 3}{\smcite{ABB+15}} &
\cc{2, 3}{\smcite{GJPW15}} &
\cc{3, 6}{\smcite{ABB+15}} & \cc{4*, 6}{\smcite{ABB+15}} \\ \hline

\raisebox{-\height}{$R$} & \ct{1, 1}{$\oplus$} & \ct{1, 1}{$\oplus$} & \cellcolor{darkgray} & \ct{2, 2}{$\wedge\circ\vee$} &
\ct{2*, 3}{$\wedge\circ\vee$} & \ct{2*, 3}{$\wedge\circ\vee$} & \cc{1.5, 3}{\smcite{ABB+15}} &
\cc{2, 3}{\smcite{GJPW15}} &
\co{2.5, 6}{\hyperref[thm:RvsQ]{Th.~\ref*{thm:RvsQ}}} & \cc{4*, 6}{\smcite{ABB+15}} \\ \hline

\raisebox{-\height}{$C$} & \ct{1, 1}{$\oplus$} & \ct{1, 1}{$\oplus$} & \ct{1, 2}{$\oplus$} & \cellcolor{darkgray} &
\ct{2, 2}{\smcite{GSS13}} & \ct{2, 2}{\smcite{GSS13}} & \ct{1.1527, 3}{\smcite{Amb13}}  & \ct{$\log_3 6$, 3}{\smcite{NW95}} & \ct{2, 4}{$\wedge$} & \ct{2, 4}{$\wedge$} \\ \hline

\raisebox{-\height}{$\RC$} & \ct{1, 1}{$\oplus$} & \ct{1, 1}{$\oplus$} & \ct{1, 1}{$\oplus$} & \ct{1, 1}{$\oplus$} & \cellcolor{darkgray} & \ct{1.5, 2}{\smcite{GSS13}} & \ct{1.1527, 2}{\smcite{Amb13}} & \ct{$\log_3 6$, 2}{\smcite{NW95}} & \ct{2, 2}{$\wedge$} & \ct{2, 2}{$\wedge$} \\  \hline

\raisebox{-\height}{$\bs$} & \ct{1, 1}{$\oplus$} & \ct{1, 1}{$\oplus$} & \ct{1, 1}{$\oplus$} & \ct{1, 1}{$\oplus$} & \ct{1, 1}{$\oplus$} & \cellcolor{darkgray} & \ct{1.1527, 2}{\smcite{Amb13}} & \ct{$\log_3 6$, 2}{\smcite{NW95}} & \ct{2, 2}{$\wedge$} & \ct{2, 2}{$\wedge$} \\  \hline

\raisebox{-\height}{$Q_E$} &  \ct{1, 1}{$\oplus$} & \ct{1.3267, 2}{$\bar{\wedge}$-tree} & \ct{1.3267, 3}{$\bar{\wedge}$-tree} & \ct{2, 2}{$\wedge\circ\vee$} & \ct{2*, 3}{$\wedge\circ\vee$} & \ct{2*, 3}{$\wedge\circ\vee$} & \cellcolor{darkgray} &
\co{2, 3}{\hyperref[thm:Qvsdeg]{Th.~\ref*{thm:Qvsdeg}}} &
\ct{2, 6}{$\wedge$} & \co{4*, 6}{\hyperref[thm:Qvsadeg]{Th.~\ref*{thm:Qvsadeg}}} \\ \hline

\raisebox{-\height}{$\deg$} &  \ct{1, 1}{$\oplus$} & \ct{1.3267, 2}{$\bar{\wedge}$-tree} & \ct{1.3267, 3}{$\bar{\wedge}$-tree} & \ct{2, 2}{$\wedge\circ\vee$} & \ct{2*, 3}{$\wedge\circ\vee$} & \ct{2*, 3}{$\wedge\circ\vee$} & \ct{1, 1}{$\oplus$}& \cellcolor{darkgray} &
\ct{2, 6}{$\wedge$} & \ct{2, 6}{$\wedge$} \\ \hline

\raisebox{-\height}{$Q$} & \ct{1, 1}{$\oplus$} & \ct{1, 1}{$\oplus$} & \ct{1, 1}{$\oplus$} &
\co{2, 2}{\hyperref[thm:QvsC]{Th.~\ref*{thm:QvsC}}} &
\co{2*, 3}{\hyperref[thm:QvsC]{Th.~\ref*{thm:QvsC}}} &
\co{2*, 3}{\hyperref[thm:QvsC]{Th.~\ref*{thm:QvsC}}} & \ct{1, 1}{$\oplus$} &
\co{2, 3}{\hyperref[thm:Qvsdeg]{Th.~\ref*{thm:Qvsdeg}}} & \cellcolor{darkgray} &
\co{4*, 6}{\hyperref[thm:Qvsadeg]{Th.~\ref*{thm:Qvsadeg}}} \\ \hline

\raisebox{-\height}{$\adeg$} & \ct{1, 1}{$\oplus$} & \ct{1, 1}{$\oplus$} & \ct{1, 1}{$\oplus$} &
\ct{7/6, 2}{$\wedge \circ \textsc{Ed}$} &
\ct{7/6, 3}{$\wedge \circ \textsc{Ed}$} &
\ct{7/6, 3}{$\wedge \circ \textsc{Ed}$} &
\ct{1, 1}{$\oplus$} & \ct{1, 1}{$\oplus$}  & \ct{1, 1}{$\oplus$} & \cellcolor{darkgray}  \\ %\hline

\end{tabular}
\label{tab:sep}
\end{center}
An entry $a, b$ in the row $M_1$ and column $M_2$ roughly\footnote{More precisely it means $a \leq \crit(M_1,M_2) \leq b$, where we define $\crit(\cdot,\cdot)$ in \sec{prelim}.}
means $M_1(f) = \tO(M_2(f)^b)$ for all total $f$ and there exists a total $f$ with $M_1(f) = \tOmega(M_2(f)^a)$. \
Each cell contains a citation or a description of a separating function, where $\oplus = \PARITY$, $\wedge = \AND$, $\wedge\circ\vee = \ANDOR$, $\bar{\wedge}$-tree is the balanced \NAND-tree function \cite{SW86,San95}, and $\wedge \circ \textsc{Ed}$ is the $\AND$ function composed with Element Distinctness.\footnote{%
The Element Distinctness function accepts a list of $n$ numbers in $[n]$ as input and asks if any two of them are equal. \
The certificate complexity of the $\AND_n$ composed with element distinctness on $n$ bits is $\tO(n)$, by composing certificates for the two functions. \ However, its approximate degree is $\tOmega(n^{7/6})$, which follows from noting that its (negative) one-sided approximate degree is $\Omega(n^{2/3})$~\cite{BT15} and that composition with the $\AND_n$ function raises its approximate degree by a factor of $\sqrt{n}$ \cite{BT15,She13}. \ (We thank Mark Bun for outlining this proof.)%
}
Entries followed by a star (e.g., 2*) correspond to separations that are optimal if $D(f)=O(\bs(f)^2)$. \
Separations colored \colorbox{open!70}{red} are
new to this work.
Separations colored \colorbox{conj!70}{gray} are separations
from recent papers that we reprove in this work.
\end{minipage}
\end{table}

\subsection{Overview of techniques}

\subsubsection*{Cheat sheet technique and randomized versus quantum query complexity}
While we know large partial function separations between randomized and quantum query complexity, such as
Simon's problem \cite{Sim97} or
the Forrelation problem \cite{AA15} that satisfies $R(g) = \tOmega(\sqrt{n})$ and $Q(g)=1$, it is unclear how to make these functions total. \ We could simply define $g$ to be $0$ on inputs not in the domain, but then the quantum algorithm would need to be able to decide whether an input is in the domain.

To solve this problem, we compose the Forrelation problem $g$ with a total function $h=\ANDOR$ on $n^2$ bits, to obtain a partial function $f = g \circ h$ on $n^3$ bits that inherits the properties of both $g$ and $h$. \ Since $\ANDOR$ has low certificate complexity, it is easier to certify that an input lies in the domain of $g \circ h$, since we only need to certify the outputs of the $n$ different $h$ gates.

Since $R(\ANDOR_{n^2}) = \Omega(n^2)$, it can be shown that $R(f) = R(g\circ h) = \Omega(n^{2.5})$. \ However, since $Q(\ANDOR_{n^2}) = O(n)$, we have $Q(f) = O(n)$. \ Now $f$ is still a partial function, but there is a small certificate for the input being in the domain. \ If we could ensure that only the quantum algorithm had access to this certificate, but the randomized algorithm did not, we would have our power $2.5$ separation.

To achieve this, we hide a ``cheat sheet'' in the input that contains all the information the quantum algorithm needs. \ We store it in a vast array of size $n^{10}$ of potential cheat sheets, which cannot be searched quickly by brute force by any algorithm, classical or quantum. \ Thus, we must provide the quantum algorithm the address of the correct cheat sheet. \ But what information do we have that is known only to the quantum algorithm? \ The value of $f$ on the input! \ While one bit does not help much, we can use $10\log n$ copies of $f$ acting on $10\log n$ different inputs. \ The outputs to these $10\log n$ problems can index into an array of size $n^{10}$, and can be determined by the quantum algorithm using only $O(n\log n)$ queries. \ At this index we will store the cheat sheet, which contains certificates for all $10\log n$ instances of $f$, convincing us that all the inputs lie in the domain. \ Our cheat sheet function now evaluates to $1$ if and only if all $10\log n$ instances of $f$ lie in the domain, \emph{and} the cell in the array pointed to by the outputs of these $10\log n$ instances of $f$ contains a cheat sheet certifying that these inputs lie in the domain.

\subsubsection*{Quantum query complexity versus certificate complexity}
Consider the $\ksum$ problem, in which we are given $n$ numbers and have to decide if any $k$ of them sum to $0 \pmod M$ for some $M$. \ For large $k$, the quantum query complexity of $\ksum$ is nearly linear in $n$ \cite{BS13}. \ While it is easy to certify $1$-inputs by showing any $k$ elements that sum to $0$, the $0$-inputs are difficult to certify and hence the certificate complexity is also linear in $n$.

Building on $\ksum$, we define a new function that we call \textsc{Block} $\ksum$, whose quantum query complexity and certificate complexity are both linear in $n$. \ However, \textsc{Block} $\ksum$ has a curious property:
for both $1$-inputs and $0$-inputs, the certificates themselves consist almost exclusively of input bits set to $1$. \
This means that if we compose this function with $\ksum$, the composed function on $n^2$ bits has certificates of size $\tO(n)$ because any certificate of \textsc{Block} $\ksum$ consists almost entirely of $1$s, which correspond to $\ksum$ instances that output $1$, which are easy to certify. \ On the other hand, the quantum query complexity of the composed function, which we call $\BKK$ for \textsc{Block} $\ksum$ of $\ksum$, is the product of the quantum query complexities of individual functions. \ Thus the certificate complexity of $\BKK$ is $\tO(n)$, but its quantum query complexity is $\tOmega(n^2)$.

\subsubsection*{Quantum query complexity versus polynomial degree}
We show that plugging any function that achieves $Q(f) = \tOmega(n^2)$ and $C(f)= \tO(n)$ into the cheat sheet framework yields a function with $Q(f) = \tOmega(n^2)$ and $\deg(f)= \tO(n)$. \ The quantum lower bound uses the hybrid method \cite{BBBV97} and the recent strong direct product theorem for quantum query complexity \cite{LR13}. \
The degree upper bound holds because for every potential cheat sheet location, there exists a degree $\tO(n)$ polynomial that checks if this location is the one pointed to by the given input. \ And since the input can point to at most one cheat sheet, the sum of these polynomials equals the function $f$.

To achieve a fourth power separation between $Q(f)$ and $\adeg(f)$, we need to check whether a cheat sheet is valid using a polynomial of degree $\tO(\sqrt{C(f)})$. \ Some certificates can be checked by Grover's algorithm in time $\smash{\tO}(\sqrt{C(f)})$, which would yield an approximating polynomial of similar degree, but this is not true for all functions $f$. \ To remedy this, we construct a new function based on composing $\BKK$ with itself recursively $\log n$ times, to obtain a new function we call $\RBKK$, and show that certificates for $\RBKK$ can be checked quickly by a quantum algorithm.

%%%%%%%%%%%%%%%%%%%%%%%%%%%%%%%%%%%%%%%%%%%%%%%%%%%%%%%%%%%%%%%%%%%%%%%%%%%%%%
\section{Preliminaries}
\label{sec:prelim}

For any positive integer $n$, let $[n] := \{1,\ldots,n\}$. \
We use $f(n) = \tO(g(n))$ to mean there exists a constant $k$ such that $f(n) = O(g(n) \log^k n)$. \
For example, $f(n) = \tO(1)$ means $f(n) = O(\log^k n)$ for some constant $k$. \
Similarly, $f(n) = \tOmega(g(n))$ denotes $f(n) = \Omega(g(n)/\log^k n)$ for some constant $k$. \
Finally, $f(n) =  \tTheta(g(n))$ means $f(n) = \tO(g(n))$ and $f(n) = \tOmega(g(n))$.

\subsection{Boolean functions}

Let $f$ be function from a domain $D \subseteq \Sigma^n$ to $\B$, where $\Sigma$ is some finite set. \
Usually $\Sigma=\B$, and we assume this for simplicity below, although most complexity measures are easily generalized to larger input alphabets. \
A function $f$ is called a total function if $D = \Sigma^n$. Otherwise, we say $f$ is a partial function or a promise problem and refer to $D$ as the promised set or  the promise.

\setlength{\intextsep}{0pt}%
\setlength{\columnsep}{10pt}%
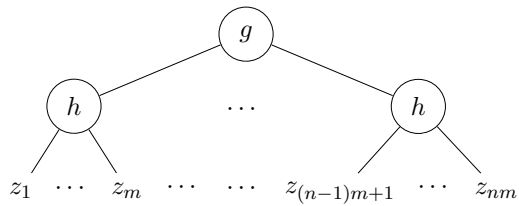
\begin{wrapfigure}{r}{0.45\textwidth}
\vspace{0.25em}
\centering
   \begin{tikzpicture}[-, scale=.9,
    level distance = 3em,
    sibling distance = 2pt,
   every internal node/.style={vert}]
    \tikzstyle{vert}=[circle,draw=black,minimum size=0.8cm,inner sep=0pt]
    \tikzset{edge from parent/.style=
        {draw, edge from parent path={(\tikzparentnode) -- (\tikzchildnode)}}}
    \Tree
    [.{$g$}
        [.{$h$}
            {$z_1$}
				\edge[draw=white];{$\cdots$}
				{$z_m$}
        ]
		\edge[draw=white];
        [.\node[draw=white]{};
            \edge[draw=white];{$\cdots$}
        ]
		\edge[draw=white];
        [.\node[draw=white]{$\cdots$};
            \edge[draw=white];{$\cdots$}
        ]
        [.{$h$}
            {$z_{(n-1)m+1}$}
				\edge[draw=white];{$\cdots$}
				{$z_{nm}$}
        ]
    ]
    \end{tikzpicture}
\vspace{-1em}
\caption{Composed function $g \circ h$.\label{fig:composed}}
\vspace{0em}
\end{wrapfigure}

For any two (possibly partial) Boolean functions $g:G \to \B$, where $G \subseteq \B^n$, and $h:H\to \B$, where $H \subseteq \B^m$, we can define the composed function $g\circ h: D \to \B$, where $D \subseteq \B^{nm}$, as follows. \ For an input $(z_1, z_2, \ldots, z_{nm}) \in \B^{nm}$, we define $g \circ h$ as follows (as depicted in \fig{composed}):
\begin{equation}
g\circ h(z) := g(h(z_1,\ldots,z_m),h(z_{m+1},\ldots,z_{2m}),\ldots,h(z_{(n-1)m+1},\ldots,z_{nm})),
\end{equation}
where the function $g \circ h$ is only defined on $z$ where the inputs to the $h$ gates lie in the domain $H$ and the $n$-bit input to $g$,  $(h(z_1,\ldots,z_m),h(z_{m+1},\ldots,z_{2m}),\ldots,h(z_{(n-1)m+1},\ldots,z_{nm}))$, lies in $G$.

Some Boolean functions that we use in this paper are $\AND_n$ and $\OR_n$, the $\AND$ and $\OR$ functions respectively on $n$ bits, and $\ANDOR_{n^2}:\B^{n^2} \to \B$, defined as $\AND_n \circ \OR_n$.

\subsection{Complexity measures}
\label{sec:MEASURES}

Formal definitions of most measures introduced here may be found in the survey on query complexity by Buhrman and de Wolf \cite{BdW02} and in the recent paper of Ambainis et al.~\cite{ABB+15}.

In the model of query complexity, algorithms may only access the input by asking queries to an oracle that knows the input $x \in D$. \ The queries are of the form ``what is $x_i$?'' for some $i \in [n]$ and the oracle responds with the value of $x_i \in \Sigma$. \
The goal is to minimize the number of queries made to compute $f$.

We use $D(f)$ to denote the deterministic query complexity of computing $f$, the minimum number of queries that have to be made by a deterministic algorithm that outputs $f(x)$ on every input $x$. \
We use $R(f)$ to denote bounded-error randomized query complexity, the minimum number of queries made by a randomized algorithm that outputs $f(x)$ on input $x$ with probability at least $2/3$. \
We use $R_0(f)$ to denote zero-error randomized query complexity, the minimum number of queries made by a randomized algorithm that must output either the correct answer, $f(x)$, on input $x$ or output $*$, indicating that it does not know the answer. \
However, the probability (over the internal randomness of the algorithm) that it outputs $*$ should be at most $1/2$ on any $x \in D$.

Similarly, we study quantum analogues of these measures. \
The quantum analogue of $D(f)$ is exact quantum query complexity, denoted $Q_E(f)$, the minimum number of queries made by a quantum algorithm that outputs $f(x)$ on every input $x\in D$ with probability 1. \
The quantum analogue of $R(f)$ is bounded-error quantum query complexity, denoted $Q(f)$, where the algorithm is only required to be correct with probability at least $2/3$.

We also study some measures related to the complexity of certifying the output. \
For a function $f$, let $C_1(f)$ denote the minimum number of bits of a $1$-input, i.e., an input with $f(x) = 1$, that have to be revealed to convince a verifier that $f(x) = 1$. \
Alternately, imagine the verifier is allowed to interact with an untrustworthy prover who claims that $f(x) = 1$, and is allowed to provide an arbitrarily long witness for this. \
After receiving the witness, the verifier makes at most $C_1(f)$ queries to the input and decides whether to accept that $f(x)=1$ or not. \ We require that if $f(x)=1$, then there is some witness that makes the verifier accept and if $f(x) = 0$ then no witness makes the verifier accept. \
Similarly we define $C_0(f)$ as the analogous notion for $0$-inputs. \
Finally the certificate complexity of $f$, or $C(f)$, is defined as $\max\{C_0(f),C_1(f)\}$. \
Similarly, if the verifier is randomized and only needs to accept valid proofs and reject invalid proofs with probability greater than $2/3$, we get the bounded-error randomized analogues $\RC_0(f)$, $\RC_1(f)$, and $\RC(f) := \max \{\RC_0(f), \RC_1(f)\}$ \cite{Aar06}.

We also study a combinatorial measure of a function called block sensitivity, denoted $\bs(f)$. \
For a function $f$ and input $x$, a block of bits $B \subseteq [n]$ is said to be sensitive for $x$ if the input obtained by complementing all the bits $x_i$ for $i \in B$ yields an input $y \in D$ with $f(x)\neq f(y)$. \
The block sensitivity of $f$ on input $x$ is the maximum number of disjoint blocks sensitive for $x$. \
The block sensitivity of $f$ is the maximum block sensitivity over all inputs  $x\in D$.

Lastly we study measures related to representing a function with a polynomial with real coefficients. \
The (exact) degree of $f$, denoted $\deg(f)$, is the minimum degree of a real polynomial $p$ over the variables $x_1,x_2,\ldots,x_n$, such that $p(x_1,x_2,\ldots,x_n) = f(x)$ for all $x \in D$. \
Similarly, the approximate degree of $f$, denoted $\adeg(f)$, is the minimum degree of a polynomial that approximates the value of $f(x)$ on every input $x\in D$, i.e., for all $x\in D$ the polynomial satisfies $|p(x)-f(x)|\leq 1/3$.

All the measures introduced in this section lie between $0$ and $n$. \
For most measures this follows because the algorithm may simply query all $n$ input variables and compute $f(x)$. \
Also, any Boolean function $f$ can be exactly represented by a real  polynomial of degree at most $n$.

\subsection{Known relations and separations}

We now summarize the best known relations between the complexity measures studied in this paper. \
\fig{rel} depicts all known relations of the type $M_1(f) = O(M_2(f))$ for complexity measures $M_1(f)$ and $M_2(f)$. \ 
An upward line from $M_1$ to $M_2$ in \fig{rel} indicates $M_1(f) = O(M_2(f))$ for all (partial or total) Boolean functions $f$. \ 
Most of these relations follow straightforwardly from definitions. \ For example $R_0(f) \leq D(f)$ because any deterministic algorithm is also a zero-error randomized algorithm. \ The relationships that do not follow from such observations are $\bs(f) = O(\RC(f))$ \cite{Aar06}, $\deg(f) =O(Q_E(f))$ \cite{BBC+01}, and $\adeg(f) = O(Q(f))$ \cite{BBC+01}.

\setlength{\intextsep}{0pt}%
\setlength{\columnsep}{15pt}%
\begin{wrapfigure}{r}{0.29\textwidth}
\vspace{-1em}
  \begin{tikzpicture}[x=1cm,y=.9cm]

     \node (D) at(2,5){$D$};
     \node (Rz) at(1,4){$R_0$};
     \node (QE) at(3,4){$Q_E$};
     \node (C) at(0,3){$C$};
     \node (R) at(2,3){$R$};
     \node (RC) at(1,2){$\RC$};
     \node (bs) at(1,0.9){$\bs$};
     \node (deg) at(4,3){$\deg$};
     \node (Q) at(3,2){$Q$};
     \node (adeg) at(3,0.9){$\adeg$};

     \path[-] (Rz) edge (D);
     \path[-] (QE) edge  (D);
     \path[-] (C) edge (Rz);
     \path[-] (R) edge (Rz);
     \path[-] (RC) edge (C);
     \path[-] (RC) edge (R);
     \path[-] (bs) edge (RC);
     \path[-] (deg) edge (QE);
     \path[-] (Q) edge (QE);
     \path[-] (Q) edge (R);
     \path[-] (adeg) edge (Q);
     \path[-] (adeg) edge (deg);
  \end{tikzpicture}
\vspace{-0.9em}
    \caption{Relations between complexity measures.\label{fig:rel}}
\vspace{-1.5em}
\end{wrapfigure}
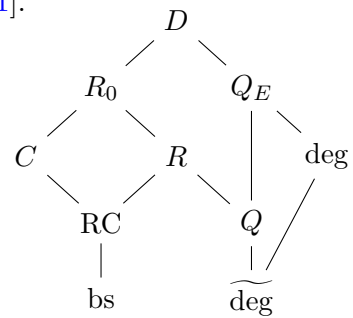

All other known relationships between these measures follow by combining the relationships in \fig{rel} with the following relationships that hold for all total Boolean functions:
\begin{itemize}[noitemsep,topsep=4pt]
\item $C(f) \leq \bs(f)^2$ \cite{Nis91}
\item $D(f) \leq C(f)\bs(f)$ \cite{BBC+01}
\item $D(f)\leq \deg(f)^3$ \cite{Mid04,Tal13}
\item $\RC(f) = O(\adeg(f)^2)$ \cite{KT13}
\item $R_0(f) = O(R(f)^2 \log R(f))$ \cite{KT13}
\end{itemize}

We now turn to the best known separations between these measures. \
To conveniently express these results, we use a notion called the \emph{critical exponent}, as defined by \cite{GSS13}. \
The the critical exponent for a measure $M_1$ relative to a measure $M_2$, denoted $\crit(M_1,M_2)$, is the infimum over all $r$ such that the relation $M_1(f) = O(M_2(f)^r)$ holds for all total functions $f$.
For example, because $D(f) \leq C(f)^2$ for all total $f$, we have $\crit(D,C) \leq 2$. \
Furthermore, since we also know that the $\ANDOR$ function on $k^2$ bits satisfies $D(f)=k^2$ and $C(f)=k$, we have $\crit(D,C)=2$. \
Note that the critical exponent between $M_1$ and $M_2$ is $2$ even if we have $M_1(f) = O(M_2(f)^2 \log M_2(f))$, or more generally if $M_1(f) \leq M_2(f)^{2+o(1)}$.

\tab{sep} lists the best known separations between these measures. \ A cell in the table has the form $a, b$, where $a \leq b$ are the best known lower and upper bounds on the critical exponent for the row measure relative to the column measure. \
The cell also contains a citation for the function that provides the lower bound. \
For example, in the cell corresponding to row $D$ and column $\deg$, we have the entry $2, 3$, which means $2 \leq \crit(D,\deg) \leq 3$, and a function achieving the separation appears in \cite{GPW15}.

%%%%%%%%%%%%%%%%%%%%%%%%%%%%%%%%%%%%%%%%%%%%%%%%%%%%%%%%%%%%%%%%%%%%%%%%%%%%%%
\section{Randomized versus quantum query complexity}
\label{sec:cheatsheet}

In this section we show a power $2.5$ separation between bounded-error randomized query complexity, $R(f)$, and bounded-error quantum query complexity, $Q(f)$. \ In \sec{intuition} we motivate the cheat sheet framework and provide a sketch of the separation. \ In \sec{implementation} we formally prove the separation.

\subsection{Intuition}
\label{sec:intuition}

We begin with the best known separation between randomized and quantum query complexity  for \emph{partial} functions, which is currently  $\tOmega(\sqrt{n})$ versus $1$ provided by the Forrelation problem \cite{AA15}. \
(Note that Simon's problem also provides a similar separation up to log factors \cite{Sim97}.) \
Let $g$ be the Forrelation function on $D \subseteq \B^n$ with $R(g) = \tOmega(\sqrt{n})$ and $Q(g) = 1$, although any function with these query complexities (up to log factors) would do.

One way to make $g$ total would be  to define it to be $0$ on the rest of the domain $\B^n \setminus D$. \
But then it is unclear if the new function $g$ has low quantum query complexity, since the quantum algorithm would have to test whether the input lies in the promised set $D$. \
In general since partial function separations often require a stringent promise on the input, we may expect that it is necessary to examine most input bits to ascertain that $x \in D$.

Indeed, it is not even clear how to \emph{certify} that $x \in D$ for our function $g$. \
On the other hand, the domain certification problem is trivial for total functions. \
So we can compose $g$ with a total function $h$ to obtain some of the desirable properties of both. \
We can certify that an input to $g\circ h$ lies in its domain by certifying the outputs of all instances of $h$. \
This means we want $h$ to have low certificate complexity, and since we will use $g\circ h$ to separate randomized and quantum query complexity, we would also like $h$ to have $Q(h)$ smaller than $R(h)$. \
A function that fits the bill perfectly is the $\ANDOR_{m^2}$ function that satisfies $R(h) = \Omega(m^2)$, $C(h) = m$, and $Q(h) = O(m)$.

We can now compute the various complexities of $f := g\circ h$. \ First we have $Q(f) = \tO(m)$, by composing algorithms for $g$ and $h$. \
There also exists a certificate of size $\tO(nm)$ that proves that the input satisfies the promise, since we can certify the outputs of the $n$ different $h$ gates using certificates of size $C(h)=m$. \
Also, given query access to this certificate of size $\tO(nm)$, a quantum algorithm can check the certificate's validity using Grover search in $\tO(\sqrt{nm})$ queries. \ Hence a quantum algorithm with access to a certificate can solve the problem with $\tO(m+\sqrt{nm})$ queries. \ On the other hand, it can be shown that $R(f) = R(g \circ h) = \Omega(R(g)R(h)) =  \tOmega(\sqrt{n}m^2)$.

Now if we set $m=n$, then we see that $Q(f) = \tO(n)$ and $R(f) = \tOmega(n^{2.5})$, but $f$ is still a partial function. \
However, $f$ has the desirable property that a quantum algorithm given query access to a certificate can decide whether the input satisfies the promise using $\tO(n)$ queries. \
But we cannot simply append the certificate to the input as we do not want the randomized algorithm to be able to use it. \ It is, however, acceptable if the randomized algorithm finds the certificate after it spends $R(f)$ queries, since that is the lower bound we want to prove.

What we would like is to hide the certificate somewhere in the input where only the quantum algorithm can find it. \
What information do we have that is only known to the quantum algorithm and remains unknown to the randomized algorithm unless it spends $R(f)$ queries? \
Clearly the value of $f$ on the input has this property.

While one bit does not help much, we can obtain additional bits of this kind by using (say) $10 \log n$ copies of $f$. \
Now the answer string to these problems can address an array of size $n^{10}$, much larger than can be brute-force searched by a randomized algorithm in a reasonable amount of time. \
Furthermore, this address can be found by a quantum algorithm using only $\tO(Q(f))$ queries. \
At this location we will store a cheat sheet: a collection of certificates for all $10 \log n$ instances of $f$.

Now the new function, which we call $f_\CS$ (shown in \fig{cs}), evaluates to $1$ if and only if the $10\log n$ inputs are in the domain of $f$, and the outputs of the $10\log n$ copies of $f$ point to a valid cheat sheet, i.e., a set of valid certificates for the $10\log n$ copies of $f$. \
This construction ensures that the quantum algorithm can find the cheat sheet using $\tO(Q(f))=\tO(n)$ queries, and then verify it using $\tO(n)$ queries, which gives $Q(f_\CS) = \tO(n)$.
On the other hand, we intuitively expect that the randomized algorithm cannot even find the cheat sheet unless it computes $f$ by spending at least  $R(f)=\tOmega(n^{2.5})$ queries, which gives us the desired separation.

\subsection{Implementation}
\label{sec:implementation}

In this section we prove \thm{RvsQ}, restated for convenience:

\RvsQ*

Let $g$ be the Forrelation function on $D \subseteq \B^n$ with $R(g) = \tOmega(\sqrt{n})$ and $Q(g) = O(1)$. \
Let $h:\B^{m^2} \to \B$ be the $\ANDOR_{m^2}$ function, defined as $\AND_m \circ \OR_m$. This function satisfies $R(h) = \Omega(m^2)$, which can be shown using a variety of methods: it follows from the partition bound \cite[Theorem 4]{JK10} and also from our more general \thm{Rcomposition}. \ We also have $C(h) = m$, since one $1$ from each $\OR$ gate is a valid $1$-certificate and an $\OR$ gate with all zero inputs is a valid $0$-certificate. \ Lastly, $Q(h) = \tO(m)$ follows from simply composing quantum algorithms for $\AND$ and $\OR$, and indeed $Q(h) = O(m)$ \cite{HMdW03}.

Let $f = g\circ h$ be a partial function on $nm^2$ bits. \ We have $Q(f) = O(m)$ by composition, since quantum algorithms can be composed in general without losing a log factor due to error reduction \cite{Rei11,LMR+11}. \
There also exists a certificate of size $\tO(nm)$ that proves that the input satisfies the promise, since we can certify the outputs of the $n$ different $h$ gates using $C(h)=m$ pointers to the relevant input bits. \
Since each pointer uses $O(\log n)$ bits, the certificate is of size $\tO(nm)$. \ Also note that a quantum algorithm with \emph{query access} to this certificate can check its validity using Grover search in $\tO(\sqrt{nm})$ queries.

Lastly, we claim that $R(f) = R(g \circ h) = \Omega(R(g)R(h)) =  \tOmega(\sqrt{n}m^2)$. \ While a general composition theorem for  bounded-error or zero-error randomized query complexity is unknown, we can show such a result when the inner function is the $\OR$ function.

\begin{theorem}
\label{thm:Rcomposition}
Let $f:D\to\B$ be a partial function, where $D \subseteq \B^n$.
Then $R(f\circ\OR_m)=\Omega(mR(f))$
and $R_0(f\circ\OR_m)=\Omega(mR_0(f))$. \
Therefore $R(f \circ \AND_m \circ \OR_m) = \Omega(m^2R(f))$ and $R_0(f \circ \AND_m \circ \OR_m) = \Omega(m^2 R_0(f))$.
\end{theorem}
\begin{proof}
We prove this for bounded-error algorithms; the argument
for zero-error algorithms will be almost identical. \
Let $A$ be the best randomized algorithm for
$f\circ\OR_m$. \ We convert $A$ into an algorithm
$B$ for $f$, as follows.

Given input $x$ to $f$, $B$ generates $n$ random
inputs to $\OR_m$, denoted by $Y_1,Y_2,\ldots,Y_n$,
in the following way. \ First, each $Y_i$ is set to $0^m$,
the all-zero string. \ Next, a random entry of $Y_i$ is
chosen and replaced with a $*$ symbol. \ Finally,
for each index $i$, the $*$ in $Y_i$ will be replaced
by the bit $x_i$. \ This causes $\OR_m(Y_i)=x_i$ to be
true for all $i$.

To evaluate $f(x)$, the algorithm $B$ can simply run $A$
on the string $Y_1Y_2\dots Y_n$. \ Since
\begin{equation}f\circ\OR_m(Y_1Y_2\dots Y_n)
    =f(\OR_m(Y_1)\OR_m(Y_2)\dots\OR_m(Y_n))
    =f(x_1x_2\dots x_n)=f(x),\end{equation}
this algorithm evaluates $f(x)$ with the same error
as $A$. \ This process uses $R(f\circ\OR_m)$ queries to the input $Y$. \
However, not all of these queries require $B$ to query
$x$. In fact, $B$ only needs to query $x$ when algorithm $A$
queries a bit that was formerly a $*$ in one of the $Y_i$. \
The expected number of queries $B$ makes is therefore
the expected number of $*$ entries found by the algorithm
$A$. \ Using Markov's inequality, we can 
turn $B$ into an $R(f)$ algorithm, which uses a fixed
number of queries to calculate $f(x)$ with bounded error;
it follows that the expected number of $*$ entries found
by $A$ is at least $\Omega(R(f))$.

We can now view $A$ as a randomized algorithm that finds
$\Omega(R(f))$ star bits (in expectation) given input
in $Y_1Y_2\dots Y_n$. \ Set $Y=Y_i$ for a randomly chosen
$i$. \ The number of queries $A$ makes to $Y$ must be
exactly $R(f\circ\OR_m)/n$ in expectation. \
Let the probability that $A$ finds $k$ stars be
$p_k$, for $k=0,1,\dots,n$. \ Then
\begin{equation}\sum_{k=0}^n kp_k=\Omega(R(f)).\end{equation}
For each $k$, the probability that $A$ finds a star in $Y$
given it found $k$ stars in total is $k/n$. \ The overall
probability that $A$ finds a star in $Y$ is therefore
\begin{equation}\sum_{k=0}^n p_k(k/n)=\frac{1}{n}\sum_{k=0}^n kp_k
    =\Omega(R(f)/n).\end{equation}

In other words, $A$ makes $R(f\circ\OR_m)/n$ expected
queries to $Y$, and finds the $*$ in $Y$ with probability
$\Omega(R(f)/n)$. \ Now, finding the single $*$ in an
otherwise all-zero string is an unordered search problem;
the chance of solving this problem after $T$ queries is
at most $T/m$. \ In other words, any decision tree of height
$T$ has probability only $T/m$ of solving the problem. \
Since $A$'s querying strategy on $Y$ can be written
as a probability distribution over decision trees
with expected height $R(f\circ\OR_m)/n$, it follows
that the probability of $A$ finding the $*$ is
at most $R(f\circ\OR_m)/nm$. \ Thus we have
$R(f\circ\OR_m)/nm=\Omega(R(f)/n)$,
or $R(f\circ\OR_m)=\Omega(mR(f))$.

If $A$ is a zero-error randomized algorithm instead of
a bounded-error algorithm, the same argument follows
verbatim, except that there's no longer a need to use
Markov's inequality to conclude that the expected number
of stars found by $A$ is at least $\Omega(R_0(f))$.

The last part follows since we can show the same result for $\AND_m$, using the fact that the query complexity of a function $f(x_1,\ldots,x_n)$ equals that of $f(\overline{x_1}, \ldots, \overline{x_n})$ and the query complexity of $f$ and $\bar{f}$
is the same.
\end{proof}

\begin{figure}
\centering
\begin{tikzpicture}[scale=0.5]
    %variables {
    \pgfmathsetmacro{\aHeight}{13.5}
    \pgfmathsetmacro{\fHeight}{\aHeight - 1.5}
    \pgfmathsetmacro{\andHeight}{\fHeight - 1.5}
    \pgfmathsetmacro{\orHeight}{\andHeight - 1.5}
    \pgfmathsetmacro{\fDotsHeight}{\fHeight-0.7}
    \pgfmathsetmacro{\fnHeight}{\fDotsHeight-0.5}
    \pgfmathsetmacro{\andDotsHeight}{\andHeight-1}
    \pgfmathsetmacro{\andnHeight}{\andDotsHeight-0.5}
    \pgfmathsetmacro{\orDotsHeight}{\orHeight-1.2}
    \pgfmathsetmacro{\ornHeight}{\orDotsHeight-0.5}
    \pgfmathsetmacro{\inputHeight}{\orHeight-1.5}
    
    \pgfmathsetmacro{\DotsShift}{0.1}
    
    \pgfmathsetmacro{\aPos}{5-1.2}
    \pgfmathsetmacro{\fPos}{\aPos}
    \pgfmathsetmacro{\fDotsPos}{\fPos+\DotsShift}
    \pgfmathsetmacro{\fnPos}{\fDotsPos}
    \pgfmathsetmacro{\andOnePos}{\aPos-2}
    \pgfmathsetmacro{\andOneDotsPos}{\andOnePos+\DotsShift}
    \pgfmathsetmacro{\andOnenPos}{\andOneDotsPos}
    \pgfmathsetmacro{\andTwoPos}{\aPos+2}
    \pgfmathsetmacro{\andTwoDotsPos}{\andTwoPos+\DotsShift}
    \pgfmathsetmacro{\andTwonPos}{\andTwoDotsPos}
    \pgfmathsetmacro{\orOnePos}{\andOnePos-1}
    \pgfmathsetmacro{\orOneDotsPos}{\orOnePos+\DotsShift}
    \pgfmathsetmacro{\orOnenPos}{\orOneDotsPos}
    \pgfmathsetmacro{\orTwoPos}{\andOnePos+1}
    \pgfmathsetmacro{\orTwoDotsPos}{\orTwoPos+\DotsShift}
    \pgfmathsetmacro{\orTwonPos}{\orTwoDotsPos}
    \pgfmathsetmacro{\orThreePos}{\andTwoPos-1}
    \pgfmathsetmacro{\orThreeDotsPos}{\orThreePos+\DotsShift}
    \pgfmathsetmacro{\orThreenPos}{\orThreeDotsPos}
    \pgfmathsetmacro{\orFourPos}{\andTwoPos+1}
    \pgfmathsetmacro{\orFourDotsPos}{\orFourPos+\DotsShift}
    \pgfmathsetmacro{\orFournPos}{\orFourDotsPos}
    \pgfmathsetmacro{\inputOnePos}{\orOnePos-0.8}
    \pgfmathsetmacro{\inputTwoPos}{\orOnePos+0.8}
    \pgfmathsetmacro{\inputThreePos}{\orTwoPos-0.8}
    \pgfmathsetmacro{\inputFourPos}{\orTwoPos+0.8}
    \pgfmathsetmacro{\inputFivePos}{\orThreePos-0.8}
    \pgfmathsetmacro{\inputSixPos}{\orThreePos+0.8}
    \pgfmathsetmacro{\inputSevenPos}{\orFourPos-0.8}
    \pgfmathsetmacro{\inputEightPos}{\orFourPos+0.8}
    % }
    
    % Forrelation 1 {
    \node (a1) at (\aPos,\aHeight) {$a_1$};
    \node (f1) at (\fPos,\fHeight) {\textsc{Forrelation}};
    \node (dots11) at (\fDotsPos,\fDotsHeight) {$\dots$};
    \node (n11) at (\fnPos,\fnHeight) {$^n$};

    \node (and11) at (\andOnePos,\andHeight) {$\wedge$};
    \node (dots12) at (\andOneDotsPos,\andDotsHeight) {$\dots$};
    \node (n12) at (\andOnenPos,\andnHeight) {$^n$};
    \node (and12) at (\andTwoPos,\andHeight) {$\wedge$};
    \node (dots13) at (\andTwoDotsPos,\andDotsHeight) {$\dots$};
    \node (n13) at (\andTwonPos,\andnHeight) {$^n$};

    \node (or11) at (\orOnePos,\orHeight) {$\vee$};
    \node (dots14) at (\orOneDotsPos,\orDotsHeight) {$\dots$};
    \node (n14) at (\orOnenPos,\ornHeight) {$^n$};
    \node (or12) at (\orTwoPos,\orHeight) {$\vee$};
    \node (dots15) at (\orTwoDotsPos,\orDotsHeight) {$\dots$};
    \node (n15) at (\orTwonPos,\ornHeight) {$^n$};
    \node (or13) at (\orThreePos,\orHeight) {$\vee$};
    \node (dots16) at (\orThreeDotsPos,\orDotsHeight) {$\dots$};
    \node (n16) at (\orThreenPos,\ornHeight) {$^n$};
    \node (or14) at (\orFourPos,\orHeight) {$\vee$};
    \node (dots17) at (\orFourDotsPos,\orDotsHeight) {$\dots$};
    \node (n17) at (\orFournPos,\ornHeight) {$^n$};

    \node (input11) at (\inputOnePos,\inputHeight) {};
    \node (input12) at (\inputTwoPos,\inputHeight) {};
    \node (input13) at (\inputThreePos,\inputHeight) {};
    \node (input14) at (\inputFourPos,\inputHeight) {};
    \node (input15) at (\inputFivePos,\inputHeight) {};
    \node (input16) at (\inputSixPos,\inputHeight) {};
    \node (input17) at (\inputSevenPos,\inputHeight) {};
    \node (input18) at (\inputEightPos,\inputHeight) {};
    
    \draw (a1) -- (f1);
    \draw (f1) -- (and11) -- (or11);
    \draw (f1) -- (and12) -- (or13);
    \draw (and11) -- (or12);
    \draw (and12) -- (or14);
    \draw (or11) -- (input11);
    \draw (or11) -- (input12);
    \draw (or12) -- (input13);
    \draw (or12) -- (input14);
    \draw (or13) -- (input15);
    \draw (or13) -- (input16);
    \draw (or14) -- (input17);
    \draw (or14) -- (input18);
    % }
    
    % Reassigning Variables {
    \pgfmathsetmacro{\aPos}{13-1.2}
    \pgfmathsetmacro{\fPos}{\aPos}
    \pgfmathsetmacro{\fDotsPos}{\fPos+\DotsShift}
    \pgfmathsetmacro{\fnPos}{\fDotsPos}
    \pgfmathsetmacro{\andOnePos}{\aPos-2}
    \pgfmathsetmacro{\andOneDotsPos}{\andOnePos+\DotsShift}
    \pgfmathsetmacro{\andOnenPos}{\andOneDotsPos}
    \pgfmathsetmacro{\andTwoPos}{\aPos+2}
    \pgfmathsetmacro{\andTwoDotsPos}{\andTwoPos+\DotsShift}
    \pgfmathsetmacro{\andTwonPos}{\andTwoDotsPos}
    \pgfmathsetmacro{\orOnePos}{\andOnePos-1}
    \pgfmathsetmacro{\orOneDotsPos}{\orOnePos+\DotsShift}
    \pgfmathsetmacro{\orOnenPos}{\orOneDotsPos}
    \pgfmathsetmacro{\orTwoPos}{\andOnePos+1}
    \pgfmathsetmacro{\orTwoDotsPos}{\orTwoPos+\DotsShift}
    \pgfmathsetmacro{\orTwonPos}{\orTwoDotsPos}
    \pgfmathsetmacro{\orThreePos}{\andTwoPos-1}
    \pgfmathsetmacro{\orThreeDotsPos}{\orThreePos+\DotsShift}
    \pgfmathsetmacro{\orThreenPos}{\orThreeDotsPos}
    \pgfmathsetmacro{\orFourPos}{\andTwoPos+1}
    \pgfmathsetmacro{\orFourDotsPos}{\orFourPos+\DotsShift}
    \pgfmathsetmacro{\orFournPos}{\orFourDotsPos}
    \pgfmathsetmacro{\inputOnePos}{\orOnePos-0.8}
    \pgfmathsetmacro{\inputTwoPos}{\orOnePos+0.8}
    \pgfmathsetmacro{\inputThreePos}{\orTwoPos-0.8}
    \pgfmathsetmacro{\inputFourPos}{\orTwoPos+0.8}
    \pgfmathsetmacro{\inputFivePos}{\orThreePos-0.8}
    \pgfmathsetmacro{\inputSixPos}{\orThreePos+0.8}
    \pgfmathsetmacro{\inputSevenPos}{\orFourPos-0.8}
    \pgfmathsetmacro{\inputEightPos}{\orFourPos+0.8}
    
    % }
    
    % Forrelation 2 {
    \node (a2) at (\aPos,\aHeight) {$a_2$};
    \node (f2) at (\fPos,\fHeight) {\textsc{Forrelation}};
    \node (dots21) at (\fDotsPos,\fDotsHeight) {$\dots$};
    \node (n21) at (\fnPos,\fnHeight) {$^n$};

    \node (and21) at (\andOnePos,\andHeight) {$\wedge$};
    \node (dots22) at (\andOneDotsPos,\andDotsHeight) {$\dots$};
    \node (n22) at (\andOnenPos,\andnHeight) {$^n$};
    \node (and22) at (\andTwoPos,\andHeight) {$\wedge$};
    \node (dots23) at (\andTwoDotsPos,\andDotsHeight) {$\dots$};
    \node (n23) at (\andTwonPos,\andnHeight) {$^n$};

    \node (or21) at (\orOnePos,\orHeight) {$\vee$};
    \node (dots24) at (\orOneDotsPos,\orDotsHeight) {$\dots$};
    \node (n24) at (\orOnenPos,\ornHeight) {$^n$};
    \node (or22) at (\orTwoPos,\orHeight) {$\vee$};
    \node (dots25) at (\orTwoDotsPos,\orDotsHeight) {$\dots$};
    \node (n25) at (\orTwonPos,\ornHeight) {$^n$};
    \node (or23) at (\orThreePos,\orHeight) {$\vee$};
    \node (dots26) at (\orThreeDotsPos,\orDotsHeight) {$\dots$};
    \node (n26) at (\orThreenPos,\ornHeight) {$^n$};
    \node (or24) at (\orFourPos,\orHeight) {$\vee$};
    \node (dots27) at (\orFourDotsPos,\orDotsHeight) {$\dots$};
    \node (n27) at (\orFournPos,\ornHeight) {$^n$};
    
    \node (input21) at (\inputOnePos,\inputHeight) {};
    \node (input22) at (\inputTwoPos,\inputHeight) {};
    \node (input23) at (\inputThreePos,\inputHeight) {};
    \node (input24) at (\inputFourPos,\inputHeight) {};
    \node (input25) at (\inputFivePos,\inputHeight) {};
    \node (input26) at (\inputSixPos,\inputHeight) {};
    \node (input27) at (\inputSevenPos,\inputHeight) {};
    \node (input28) at (\inputEightPos,\inputHeight) {};
    
    \draw (a2) -- (f2);
    \draw (f2) -- (and21) -- (or21);
    \draw (f2) -- (and22) -- (or23);
    \draw (and21) -- (or22);
    \draw (and22) -- (or24);
    \draw (or21) -- (input21);
    \draw (or21) -- (input22);
    \draw (or22) -- (input23);
    \draw (or22) -- (input24);
    \draw (or23) -- (input25);
    \draw (or23) -- (input26);
    \draw (or24) -- (input27);
    \draw (or24) -- (input28);
    % }

    % Reassigning Variables {
    \pgfmathsetmacro{\aPos}{23-1.2}
    \pgfmathsetmacro{\fPos}{\aPos}
    \pgfmathsetmacro{\fDotsPos}{\fPos+\DotsShift}
    \pgfmathsetmacro{\fnPos}{\fDotsPos}
    \pgfmathsetmacro{\andOnePos}{\aPos-2}
    \pgfmathsetmacro{\andOneDotsPos}{\andOnePos+\DotsShift}
    \pgfmathsetmacro{\andOnenPos}{\andOneDotsPos}
    \pgfmathsetmacro{\andTwoPos}{\aPos+2}
    \pgfmathsetmacro{\andTwoDotsPos}{\andTwoPos+\DotsShift}
    \pgfmathsetmacro{\andTwonPos}{\andTwoDotsPos}
    \pgfmathsetmacro{\orOnePos}{\andOnePos-1}
    \pgfmathsetmacro{\orOneDotsPos}{\orOnePos+\DotsShift}
    \pgfmathsetmacro{\orOnenPos}{\orOneDotsPos}
    \pgfmathsetmacro{\orTwoPos}{\andOnePos+1}
    \pgfmathsetmacro{\orTwoDotsPos}{\orTwoPos+\DotsShift}
    \pgfmathsetmacro{\orTwonPos}{\orTwoDotsPos}
    \pgfmathsetmacro{\orThreePos}{\andTwoPos-1}
    \pgfmathsetmacro{\orThreeDotsPos}{\orThreePos+\DotsShift}
    \pgfmathsetmacro{\orThreenPos}{\orThreeDotsPos}
    \pgfmathsetmacro{\orFourPos}{\andTwoPos+1}
    \pgfmathsetmacro{\orFourDotsPos}{\orFourPos+\DotsShift}
    \pgfmathsetmacro{\orFournPos}{\orFourDotsPos}
    \pgfmathsetmacro{\inputOnePos}{\orOnePos-0.8}
    \pgfmathsetmacro{\inputTwoPos}{\orOnePos+0.8}
    \pgfmathsetmacro{\inputThreePos}{\orTwoPos-0.8}
    \pgfmathsetmacro{\inputFourPos}{\orTwoPos+0.8}
    \pgfmathsetmacro{\inputFivePos}{\orThreePos-0.8}
    \pgfmathsetmacro{\inputSixPos}{\orThreePos+0.8}
    \pgfmathsetmacro{\inputSevenPos}{\orFourPos-0.8}
    \pgfmathsetmacro{\inputEightPos}{\orFourPos+0.8}
    
    % }
    
    % Forrelation 3 {
    \node (a3) at (\aPos,\aHeight) {$a_{10\log n}$};
    \node (f3) at (\fPos,\fHeight) {\textsc{Forrelation}};
    \node (dots31) at (\fDotsPos,\fDotsHeight) {$\dots$};
    \node (n31) at (\fnPos,\fnHeight) {$^n$};

    \node (and31) at (\andOnePos,\andHeight) {$\wedge$};
    \node (dots32) at (\andOneDotsPos,\andDotsHeight) {$\dots$};
    \node (n32) at (\andOnenPos,\andnHeight) {$^n$};
    \node (and32) at (\andTwoPos,\andHeight) {$\wedge$};
    \node (dots33) at (\andTwoDotsPos,\andDotsHeight) {$\dots$};
    \node (n33) at (\andTwonPos,\andnHeight) {$^n$};

    \node (or31) at (\orOnePos,\orHeight) {$\vee$};
    \node (dots34) at (\orOneDotsPos,\orDotsHeight) {$\dots$};
    \node (n34) at (\orOnenPos,\ornHeight) {$^n$};
    \node (or32) at (\orTwoPos,\orHeight) {$\vee$};
    \node (dots35) at (\orTwoDotsPos,\orDotsHeight) {$\dots$};
    \node (n35) at (\orTwonPos,\ornHeight) {$^n$};
    \node (or33) at (\orThreePos,\orHeight) {$\vee$};
    \node (dots36) at (\orThreeDotsPos,\orDotsHeight) {$\dots$};
    \node (n36) at (\orThreenPos,\ornHeight) {$^n$};
    \node (or34) at (\orFourPos,\orHeight) {$\vee$};
    \node (dots37) at (\orFourDotsPos,\orDotsHeight) {$\dots$};
    \node (n37) at (\orFournPos,\ornHeight) {$^n$};

    \node (input31) at (\inputOnePos,\inputHeight) {};
    \node (input32) at (\inputTwoPos,\inputHeight) {};
    \node (input33) at (\inputThreePos,\inputHeight) {};
    \node (input34) at (\inputFourPos,\inputHeight) {};
    \node (input35) at (\inputFivePos,\inputHeight) {};
    \node (input36) at (\inputSixPos,\inputHeight) {};
    \node (input37) at (\inputSevenPos,\inputHeight) {};
    \node (input38) at (\inputEightPos,\inputHeight) {};
    
    \draw (a3) -- (f3);
    \draw (f3) -- (and31) -- (or31);
    \draw (f3) -- (and32) -- (or33);
    \draw (and31) -- (or32);
    \draw (and32) -- (or34);
    \draw (or31) -- (input31);
    \draw (or31) -- (input32);
    \draw (or32) -- (input33);
    \draw (or32) -- (input34);
    \draw (or33) -- (input35);
    \draw (or33) -- (input36);
    \draw (or34) -- (input37);
    \draw (or34) -- (input38);
    % }
    
    % {
    \pgfmathsetmacro{\cheatsheetHeight}{4}
    
    \node (bigdots) at (18-1.2,\fHeight) {$\dots$};
    \draw (14.5-1.2,\aHeight) ellipse (12 and 0.7);
    \draw[->] (25.3,\aHeight) .. 
        controls(29,\aHeight/2)and (20,\cheatsheetHeight+4)
        .. (20.5,\cheatsheetHeight + 1);
    
    \draw[fill=blue!50] (20,\cheatsheetHeight)
        rectangle (21,1+\cheatsheetHeight);
    \draw (20,\cheatsheetHeight) --
        (19.5,\cheatsheetHeight-4) --
        (21.5,\cheatsheetHeight-4) --
        (21,\cheatsheetHeight);
    \draw[->] (20.5,\cheatsheetHeight-3) -- (n15);
    \draw[->] (20.5,\cheatsheetHeight-2) -- (n26);
    \draw[->] (20.5,\cheatsheetHeight-1) -- (n34);
    \node (array) at (8,\cheatsheetHeight-1)
        {array of size $n^{10}$};
    \node (cell) at (25.5,\cheatsheetHeight-2)
        {\begin{tabular}{l} marked cell containing\\
        pointers to certificates\\
        proving the value of $a$
        \end{tabular}};
    
    \foreach \x in {5,...,29}
        \draw (\x,\cheatsheetHeight) rectangle 
            (\x+1,1+\cheatsheetHeight);
    % }
\end{tikzpicture}
\caption{The function $f$ that achieves a superquadratic separation between $Q(f)=\tO(n)$ and $R(f)=\tOmega(n^{2.5})$.}
\label{fig:cs}
\end{figure}
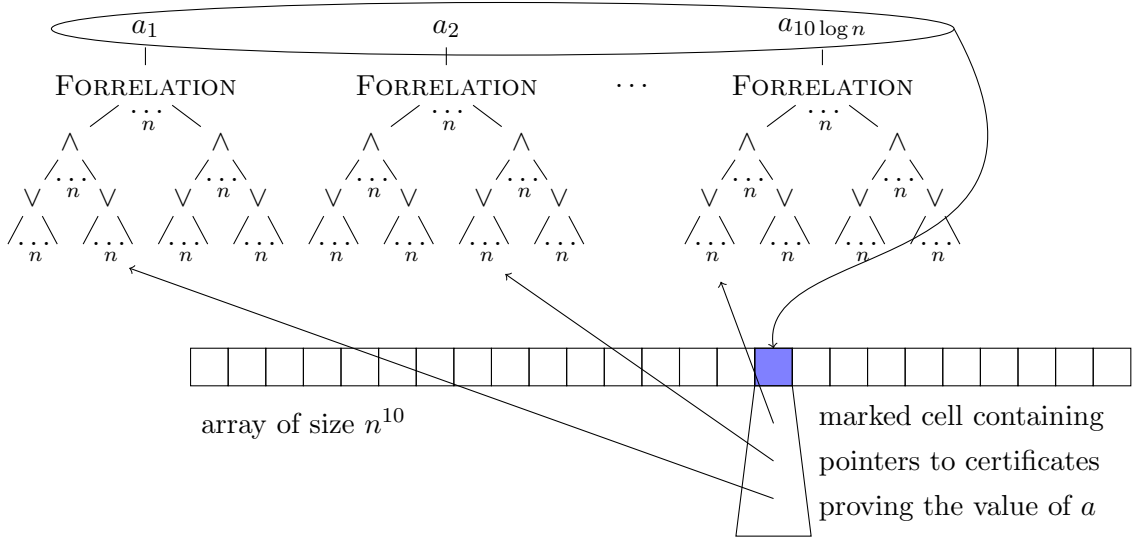

We now define the cheat sheet version of $f$, which we call $f_\CS$, depicted in \fig{cs}. \ Let the input to $f_\CS$ consist of $10 \log n$ inputs to $f$, each of size $nm^2$, followed by $n^{10}$ blocks of bits of size $\tO(mn)$ each, which we refer to as the array of cells or array of cheat sheets. \ Each cell is large enough to hold certificates for all $10\log n$ inputs to $f$ that certify for each input the output of $f$ evaluated on that input and that the promise holds.

Let us denote the input to $f_\CS$ as $z = (x^1, x^2, \ldots, x^{10\log n}, Y_1, Y_2,\ldots, Y_{n^{10}})$, where $x^i$ is an input to $f$, and the $Y_i$ are the aforementioned cells of size $\tO(mn)$. \ We define the value of $f_\CS(z)$ to be $1$ if and only if the following conditions hold:
\begin{enumerate}[topsep=4pt,noitemsep]
\item For all $i$, $x^i$ is in the domain of $f$. If this condition is satisfied, let $\ell$ be the positive integer corresponding to the binary string $(f(x^1),f(x^2),\ldots,f(x^{10\log n}))$.
\item $Y_\ell$ certifies that all $x^i$ are in the domain of $f$ and that $\ell$ equals the binary string formed by their output values, $(f(x^1),f(x^2),\ldots,f(x^{10 \log n}))$.
\end{enumerate}

We now upper bound the quantum query complexity of $f_\CS$. \ The quantum algorithm starts by assuming that the first condition of $f_\CS$ holds and simply computes $f$ on all $10\log n$ inputs, which uses $\tO(Q(f)) = \tO(m)$ queries. \
The answers to these inputs points to the cheat sheet $Y_\ell$, where $\ell$ is the integer corresponding to the binary string $(f(x^1),f(x^2),\ldots,f(x^{10\log n}))$. \
As discussed, verifying the cheat sheet of size $\tO(mn)$ requires only $\tO(\sqrt{mn})$ queries by a recursive application of Grover search. \ The algorithm outputs $1$ if and only if the verification of $Y_\ell$ succeeded. If the certificate is accepted by the algorithm, then both conditions of $f_\CS$ are satisfied and hence it is easy to see the algorithm is correct. \ Hence $Q(f_\CS) = \tO(m+\sqrt{mn})$.

We also know that $R(f) = \tOmega(\sqrt{n}m^2)$. \ We now prove that this implies $R(f_\CS) = \tOmega(R(f)) = \tOmega(\sqrt{n}m^2)$. \ Proving this completes the proof of \thm{RvsQ}, since setting $m=n$ immediately yields $Q(f_\CS) = \tO(n)$ and $R(f_\CS) = \tOmega(n^{2.5})$.

To complete the proof, we show in general that $R(f_\CS) = \tOmega(R(f))$.

\begin{lemma}
\label{lem:Rlowerbound}
Let $f:D\to\B$ be a partial function,
where $D\subseteq [M]^n$, and let $f_\CS$ be the cheat
sheet version of $f$ with $c=10\log n$ copies of $f$. \
Then $R(f_{\CS})=\Omega(R(f)/c^2)=\tOmega(R(f))$.
\end{lemma}
\begin{proof}
Let $A$ be any bounded-error randomized algorithm for evaluating $f_{\CS}$. \ We will prove that $A$ makes $\Omega(R(f)/c^2)$ queries.

We start by proving the following claim:
Let $x^1,x^2,\ldots,x^c\in\Dom(f)$ denote $c$ valid inputs to $f$ and let $z$ be an input to $f_{\CS}$ consisting of
$x^1x^2\cdots x^c$ with blank array---that is, it has $0$ in all the entries of all cheat sheets. \ Let us assume that the all zero string is not a valid cheat sheet for any input. \ This can be enforced, for example, by requiring that all cheat sheets begin with the first bit equal to $1$. \ Let $\ell$ be the binary number $f(x^1)f(x^2)\cdots f(x^c)$, and let $h_\ell$ be the
\th{\ell} cheat sheet. \ Then we claim that $A$ must query a bit of $h_\ell$ when run on $z$ with probability at least $1/3$.

This claim is easy to prove. \ Since the all zero string is an invalid cheat sheet, it follows that $f_{\CS}(z)=0$, so $A$ outputs $0$ when run on $z$ with probability at least $2/3$. \ Now if we modify $h_\ell$ in $z$ to be the correct cheat sheet for the given input $x^1,x^2,\ldots,x^c$, then we obtain a new input $z'$ with $f_{\CS}(z^\prime)=1$. \
This means $A$ outputs $1$ when run on $z^\prime$ with probability at least $2/3$. \
However, since these inputs only differ on $h_\ell$, if $A$ does not query $h_\ell$ with high probability, it cannot distinguish input $z$ from input $z'$. Since $A$ is a valid algorithm for $f_\CS$, the claim follows.

We now use the hybrid argument to prove the lemma. \
For any input $z$ with blank array,
let $p_z\in[0,1]^{2^c}$ be
the vector such that for each $i=[2^c]$,
$(p_z)_i$ is the probability
that $A$ makes a query to the \th{i} cheat sheet when
given input $z$. \ Then the $1$-norm of $p_z$ is at most
the running time of $A$, which is at most $R(f)$. \
By the claim, if $\ell$ is the relevant cheat sheet for $z$,
we have $(p_z)_\ell\geq 1/3$. \
On the other hand, since $p_z$ has
$2^c\geq n^{10} \geq R(f)^{10}$
entries that sum to at most $R(f)$,
almost all of them have value
less than, say, $R(f)^{-5}$.

Next, consider hard distributions $\mathcal{D}^0$
and $\mathcal{D}^1$ over the $0$- and $1$-inputs to $f$,
respectively. \ We pick these distributions
such that distinguishing between them
with probability at least $1/2+1/12c$ takes at least
$\Omega(R(f)/c^2)$ queries for any randomized algorithm.

For each $i\in[2^c]$,
let $q_i$ be the expectation over
$p_z$ when $z$ is made of $c$ inputs to $f$ generated from
$\mathcal{D}^i=
\mathcal{D}^{i_1}\times\mathcal{D}^{i_2}\times\dots\times\mathcal{D}^{i_c}$, together with a blank array
(here $i_j$ means the \th{j} bit of $i$
when written in binary). \
Then for all $i\in[2^c]$, we have
$(q_i)_i\geq 1/3$, and the $1$-norm of $q_i$ is at most
$R(f)$, so most entries of $q_i$ are less than $R(f)^{-5}$. \
The entries of $q_i$ can be interpreted
as the probabilities of each cheat sheet being queried
by the algorithm on an input sampled from $\mathcal{D}^i$.

Let $k\in[2^c]$ be such that $(q_0)_k<1/6$. \
Let $k_0,k_1,\dots,k_c$ be such that
$k_0=0$, $k_c=k$, and consecutive $k$'s differ by at most
one bit (when represented as binary vectors of length $c$). \
Since $(q_{k_0})_k<1/6$ and $(q_{k_c})_k\geq 1/3$,
there must be some $j\in[c]$ such that
$(q_{k_{j+1}})_k-(q_{k_j})_k>1/6c$ and $(q_{k_j})_k<1/3$.
Let $a=(q_{k_{j+1}})_k$ and let $b=(q_{k_j})_k$.

We can now use this to distinguish the product distribution
$\mathcal{D}^{k_{j+1}}$ from $\mathcal{D}^{k_j}$,
with probability of error at most $1/2-1/12c$: simply
run the algorithm $A$, output $1$ if it queried an entry
of the \th{k} cheat sheet, and otherwise output $1$ with
probability $\max(0,\frac{1-a-b}{2-a-b})$. \
If this maximum is $0$, it means we have $b>2/3$ and
$a<1/3$, so this algorithm has error at most $1/3$. \
Otherwise, it is not hard to check that the algorithm
will determine if the input is from $\mathcal{D}^{k_{j+1}}$
with probability at most $1/2-(b-a)/2<1/2-1/12c$.

Finally, note that since $k_{j+1}$ and $k_j$ differ
in only one bit, distinguishing $\mathcal{D}^{k_{j+1}}$
from $\mathcal{D}^{k_j}$ allows us to distinguish between
$\mathcal{D}^0$ and $\mathcal{D}^1$. \ This is because given
any input from $\mathcal{D}^0$ or $\mathcal{D}^1$, the
algorithm can sample other inputs to construct an input
from $\mathcal{D}^{k_{j+1}}$ or $\mathcal{D}^{k_j}$,
and then run the distinguishing algorithm. \
It follows that the running time of $A$ is at least
$\Omega(R(f)/c^2)$, by the choice of the distributions
$\mathcal{D}^0$ and $\mathcal{D}^1$.
\end{proof}

%%%%%%%%%%%%%%%%%%%%%%%%%%%%%%%%%%%%%%%%%%%%%%%%%%%%%%%%%%%%%%%%%%%%%%%%%%%%%%
\section{Quantum query complexity versus certificate complexity and degree}
\label{sec:BKK}

We now show a nearly quadratic separation between $Q(f)$ and $C(f)$, which yields a similar separation between $Q(f)$ and $\deg(f)$. \ We will also use the functions introduced in this section as building blocks to obtain the nearly \th{4} power separation between $Q(f)$ and $\adeg(f)$ proved in \sec{QUARTIC}.

\subsection{Quadratic gap with certificate complexity}

In this section we establish \thm{QvsC}, restated for convenience:

\QvsC*

Consider the $\ksum$ problem, $\ksum:[M]^n \to \B$, which asks if there are $k$ elements in the input string $x_1, x_2, \ldots, x_n \in [M]$ that sum to $0 \pmod M$. \
Belovs and \v{S}palek \cite{BS13} showed that $Q(\ksum) = \Omega(n^{k/(k+1)})$ when the alphabet $M$ has size $n^k$ and $k$ is constant. \ In \app{ksum},
we show that their proof implies a bound of $Q(\ksum)=\Omega(n^{k/(k+1)}/\sqrt{k})$ for super-constant $k$.

Now the $1$-certificate complexity of the $\ksum$ problem is $k$ (assuming it costs one query to get an element of $M$),
since it suffices to provide the $k$ elements that sum to $0 \pmod M$. \ If we take $k=\log n$, we get
$Q(\ksum)=\Omega(n/\sqrt{\log n})$ and
$C_1(\ksum)=O(\log n)$. \ Although this function is not Boolean, turning it into a Boolean function
will only incur an additional polylogarithmic loss.

While $\ksum$ does not separate quantum query complexity from certificate complexity, its $1$-certificate complexity is much smaller than its quantum query complexity. \
Composing a function with small $0$-certificates, such as the $\AND_n$ with $\ksum$ already gives a function whose quantum query complexity is larger than its certificate complexity: in this case, we have certificate complexity $\tO(n)$ and quantum query complexity $\tOmega(n^{3/2})$, which follows from the following general composition theorem~\cite{HLS07,Rei11,LMR+11,Kim12}:

\begin{theorem}[Composition theorem for quantum query complexity]
\label{thm:Qcomposition}
Let $f:D\to \B$ and $g:E \to \B$ be partial functions where $D \subseteq \B^n$ and $E \subseteq \B^m$. \
Then $Q(f \circ g) = \Theta(Q(f)Q(g))$.
\end{theorem}

To get an almost quadratic gap between $Q(f)$ and $C(f)$, we use a variant of
$\ksum$ itself as the outer function instead of
the $\AND$ function. \
From $\ksum$, we define a new Boolean function that we call \textsc{Block} $\ksum$, whose quantum query complexity is $\tTheta(n)$ and certificate complexity is also $\tTheta(n)$. \
However, although its certificate complexity is linear, the certificates consist almost exclusively of input bits set to $1$ and only $\tO(1)$ input bits set to $0$. \
This means if we compose this function with $\ksum$, the composed function has certificates of size $\tO(n)$, since the certificates of \textsc{Block} $\ksum$ are essentially composed of $1$s, which are easy to certify for $\ksum$. \ We denote this composed function $\BKK:\B^{n^2} \to \B$. \ It satisfies $C(\BKK) = \tO(n)$ and $Q(\BKK) = \tOmega(n^2)$, which yields the desired quadratic separation.

We now define the \textsc{Block $k$-sum} problem.

\begin{definition}
    Let \textsc{Block $k$-sum} be a total Boolean function
    on $n$ bits defined as follows. \
    We split the input into blocks of size
    $10k\log n$ each and say a block is \emph{balanced}
    if it has an equal number of $0$s and $1$s. \
    Let the balanced
    blocks represent numbers in an alphabet $M$ of size
    $\Omega(n^k)$. The value of the function is
    $1$ if and only if there are $k$ balanced blocks
    whose corresponding numbers sum to $0 \pmod M$ and all other
    blocks have at least as many $1$s as $0$s.
\end{definition}

We then compose this function with $\ksum$ to get the function $\BKK$.

\begin{definition}
    Let $\BKK_{n^2,k}:\B^{n^2} \to \B$ be the function
    \textsc{Block $k$-sum} on $n$ bits composed with a Boolean
    version of $\ksum$ on $n$ bits,
    and define $\BKK_{n^2}$ to be $\BKK_{n^2,k}$
    with $k=\log n$.
\end{definition}

We are now ready to establish the various complexities of $\BKK$.

\begin{theorem}\label{thm:BKK}
    For the total Boolean function $\BKK_{n^2,k}:\B^{n^2} \to \B$, we have
\begin{equation}
C(\BKK_{n^2,k})=O(k^2 n\log n) \quad \text{and} \quad
Q(\BKK_{n^2,k})=\Omega
        \Bigl(\frac{n^{2-2/(k+1)}}{k^3\log^2 n}\Bigr).
\end{equation}
\end{theorem}

\begin{proof}
We start by analyzing the certificates
of \textsc{Block $k$-sum}. \
The key property we need
is that every input of \textsc{Block $k$-sum}
has a certificate that uses very few $0$s, but can use a large
number of $1$s. \ To see this,
note that we can certify a $1$-input by showing
the $k$ balanced blocks that sum to $0$
which requires $O(k^2\log n)$ $0$s,
and all the $1$s in every other block. \
There are two kinds of $0$ inputs to certify:
A $0$-input that has a block with more $0$s than $1$s
can be certified by providing that block, which only uses $O(k\log n)$ $0$s. \
A $0$-input in which all blocks have at least
as many $1$s as $0$s can be certified by
providing all the $1$s: this provides the number
represented by each block if it were balanced
(though it does not prove the block is actually balanced),
which is enough to check that no $k$ of them sum to zero. \
In conclusion, \textsc{Block $k$-sum} can always be
certified by providing $O(n)$ $1$s and $O(k^2\log n)$
$0$s.

We now analyze the certificate complexity
of $\BKK_{n^2,k}$. \
For each input, the outer \textsc{Block $k$-sum}
has a certificate using $O(k^2\log n)$ $0$s
and $O(n)$ $1$s. \ The inner function, $\ksum$, has
$1$-certificates of size $O(k^2\log n)$ since there are
$k$ numbers to exhibit and each uses $k\log n$
bits when represented in binary, and has $0$-certificates of size $O(n)$. \
Therefore, the composed function
always has a certificate of size $O(k^2 n\log n)$. \
Hence $C(\BKK_{n^2,k})=O(k^2 n\log n)$.

The quantum query complexity
of \textsc{Block $k$-sum} is $\Omega(Q(\ksum)/k\log n)$
by a reduction from $\ksum$. \
Using the result in \app{ksum}, this is
$\Omega(n^{1-1/(k+1)}/k^{3/2}\log n)$. \
Invoking the composition theorem for quantum query complexity (\thm{Qcomposition}), we get
$Q(\BKK_{n^2,k})=\Omega\left(
\frac{n^{2-2/(k+1)}}{k^3\log^2 n}\right)$.
\end{proof}

Thus for the function $\BKK:\B^{n^2} \to \B$, defined as $\BKK = \BKK_{n^2,\log n}$, we have
    \begin{equation}
    C(\BKK)=O(n\log^3 n) =\tO(n) \quad \text{and} \quad
    Q(\BKK)=\Omega\Bigl(\frac{n^2}{\log^5 n}\Bigr) =\tOmega(n^2).
    \end{equation}
This establishes \thm{QvsC}, since $\BKK$ is a total Boolean function.

\subsection{Quadratic gap with (exact) degree}
\label{sec:Qvsdeg}

We now show how to obtain a total function that nearly quadratically separates $Q(f)$ from $\deg(f)$ using any total function that achieves a similar separation between $Q(f)$ and $C(f)$. \ This proves \thm{Qvsdeg}:

\Qvsdeg*

Let $f:\B^n \to \B$ be a total function with $Q(f) = \tOmega(n)$ and $C(f) = \tO(\sqrt{n})$, such as the $\BKK$ function introduced in the previous section. \ Let $f_\CS$ denote the cheat sheet version of $f$ (as described in \sec{cheatsheet}), created using $10\log n$ copies of $f$ that point to a cheat sheet among $n^{10}$ potential cheat sheets, where a valid cheat sheet contains certificates of size $\tO(\sqrt{n})$ for each of the $10\log n$ inputs to $f$ and the binary string corresponding to their outputs equals the location of the cheat sheet. \
In this case $f$ is a total function so the cheat sheets do not certify that the input satisfies the promise, but only the value of $f$ evaluated on the input. \
We claim that the cheat sheet version of $f$ satisfies $Q(f_\CS) = \tOmega(n)$ and $\deg(f_\CS) = \tO(\sqrt{n})$.

Let us start with the degree upper bound, $\deg(f_\CS) = \tO(\sqrt{n})$. Let $\ell \in [n^{10}]$ be a potential location of the cheat sheet. \
For any $\ell$, consider the problem of outputting $1$ if and only if $f_\CS(z)=1$ and $\ell$ is the location of the cheat sheet for the input $z$. \
Since $\ell$ is known, this can be solved by a deterministic algorithm $\mathcal{A}_\ell$ that makes $\tO(\sqrt{n})$ queries, since it can simply check if the certificate stored at cell $\ell$ in the array is valid: it can check all the $10 \log n$ certificates of size $\tO(\sqrt{n})$ for each of the $10 \log n$ instances of $f$ and then check if the outputs of $f$ evaluate to the location $\ell$. \
Since polynomial degree is at most deterministic query complexity, we can construct a representing polynomial for $\mathcal{A}_\ell$ for any location $\ell$. \ This is a polynomial $p_\ell$ of degree $\tO(\sqrt{n})$ such that $p_\ell(z)=1$ if and only if $f_\CS(z)=1$ and $\ell$ is the position of the cheat sheet on input $z$. \
Now we can simply add all the polynomials $p_\ell$ together to obtain a new polynomial $q$ of the same degree. \
We claim $q(z) = f_\CS(z)$ since if $f_\CS(z) = 0$ then certainly all the polynomials $p_\ell(z)=0$ (since none of the cheat sheets is valid) and if $f_\CS(z) = 1$ then $q(z) = 1$ because exactly one of many $p_\ell(z)$ will evaluate to $1$, the one corresponding to the location of the cheat sheet for the input $z$. \ Note that the property used here is that in a $1$-input to $f_\CS$, exactly one location serves as the correct cheat sheet, i.e., the location of the cheat sheet for a $1$-input is unique.

The claim that $Q(f_\CS) = \tOmega(n)$ seems intuitive since $Q(f) = \tOmega(n)$ and the cheat sheet version of $f$ cannot be easier than $f$ itself. \ This intuitive claim is true and we show below that $Q(f_{\CS})=\Omega(Q(f))$ in general, which completes the proof of \thm{Qvsdeg}.

To prove this general result for quantum query complexity, we will need the following
strong direct product theorem due to Lee and Roland \cite{LR13}.

\begin{theorem}
\label{thm:sdpt}
Let $f$ be a (partial) function with ${Q}_{1/4}\left(f\right)\geq T$. \
Then any $T$-query quantum algorithm that outputs the value of $f$ evaluated on $c$ independent instances has success probability at most $O(\left(3/4\right)^{c/2})$.
\end{theorem}

We now prove the lower bound on the quantum
query complexity of cheat sheet functions.

\begin{lemma}
\label{lem:Qlowerbound}
Let $f:D\to\B$ be a partial function,
where $D\subseteq [M]^n$, and let $f_\CS$ be a cheat
sheet version of $f$ with $c=10\log n$ copies of $f$. \
Then $Q(f_{\CS})=\Omega(Q(f))$.
\end{lemma}

\begin{proof}
By \thm{sdpt}, we know that given $c=10\log n$ instances of $f$, any quantum
algorithm that makes fewer than $T={Q}_{1/4}(f)$
queries will correctly decide all the
instances with probability at most
$O(\left(3/4\right)^{c/2})$.

Now suppose by way of contradiction that
${Q}_{1/4}\left(f_{\CS}\right)\leq T$. \
Then we will show how to decide $c$
copies of $f$ with probability
$\Omega\left(1/T^{2}\right)$
by making $T$ quantum queries. \ Since
\begin{equation}
\frac{1}{T^2}
\geq\frac{1}{n^{2}}
\geq\frac{1}{2^{c/5}}\gg\left(\frac{3}{4}\right)^{c/2},
\end{equation}
this is enough to contradict \thm{sdpt}.

Let $Q$ be a quantum algorithm that decides $f_{\CS}$
using at most $T$ queries. \
Then consider running $Q$ on an input
$z=(x^1,\ldots,x^c,Y_1,Y_2,\ldots,Y_{2^c})$, where the $x^i$ are in the domain of $f$
and whose cheat-sheet array $(Y_1,Y_2,\ldots,Y_{2^c})$ has been completely zeroed
out. \ We assume again that the all-zero cheat sheet is invalid for any input. \
This can always be enforced by requiring a valid cheat sheet to have the first
bit set to 1.

For all $y\in\left[2^c\right]$
and $t\in\left[ T\right]$, define $m_{y,t}$, or
\textit{the query magnitude on }$y$
\textit{at query} $t$,
to be the probability that $Q$ would be found
querying some bit in cell
$Y_y$, were we to measure in the standard basis
immediately before the \th{t} query.

For an input $z$ where we have zeroed out the cheat sheet array,
clearly $f_{\CS}(z)=0$ and hence $Q$ outputs $0$ on this input with
high probability. \ On the other hand, if we let
$\ell=f\left(x^1\right)\cdots f\left(x^c\right)\in[2^c]$,
by modifying only the \th{\ell} cell in the array, $Y_\ell$,
we could produce an input $z^{\prime}$
such that $f_{\CS}\left(z^{\prime}\right)=1$. \
From these facts, together with the standard
BBBV hybrid argument \cite{BBBV97}, it follows that
\begin{equation}
\sqrt{m_{\ell,1}}+\cdots+\sqrt{m_{\ell,T}}
    =\Omega\left(1\right).
\end{equation}
So by the Cauchy--Schwarz inequality, we have
\begin{equation}
m_{\ell,1}+\cdots+m_{\ell,T}=\Omega\left(\frac{1}{T}\right).
\end{equation}
This implies that, if we simply choose a
$t\in\left[  T\right]$ uniformly
at random, run $Q$ until the \th{t}
query with the cheat-sheet array zeroed out,
and then measure in the standard basis, we will observe
$\ell=f\left(x^1\right)\cdots f\left(x^c\right)$
with probability $\Omega\left(1/T^2\right)$. \
This completes the proof.
\end{proof}

%%%%%%%%%%%%%%%%%%%%%%%%%%%%%%%%%%%%%%%%%%%%%%%%%%%%%%%%%%%%%%%%%%%%%%%%%%%%%%
\section{Quantum query complexity versus approximate degree}
\label{sec:QUARTIC}

We now show how to obtain  a nearly quartic separation between quantum query complexity and approximate degree from a function that quadratically separates $Q(f)$ from $C(f)$. \ As in \sec{cheatsheet}, we first motivate the construction in \sec{int} and then formally give the separation in \sec{imp}. \ The main result of this section is \thm{Qvsadeg}:

\Qvsadeg*

\subsection{Intuition}
\label{sec:int}

To obtain the quartic separation with approximate degree, we could try the same approach as in the previous section for exact degree. \ Using notation from \sec{Qvsdeg}, consider the algorithm $\mathcal{A}_\ell$ that makes $\tO(\sqrt{n})$ queries. \
Instead of using a polynomial to exactly represent $\mathcal{A}_\ell$, we could try to construct more efficient approximating polynomials. \
If there were a quantum algorithm that checked the certificate quadratically faster than the deterministic $\mathcal{A}_\ell$, then we would be done. \
This is because such a quantum algorithm would yield a polynomial of degree $\tO(n^{1/4})$ that approximates $\mathcal{A}_\ell$. \ Then we could sum up these polynomials as before, with some error reduction to ensure that the error in each polynomial is small enough not to affect the sum of the polynomials.  \ This error reduction only adds an overhead of $O(\log n)$, which gives us a polynomial of degree $\tO(n^{1/4})$ that approximates $f_\CS$.

However, it is unclear if there is a quantum algorithm to check a certificate of size $\tO(\sqrt{n})$ quickly. \
Reading the certificate itself could require $\tOmega(\sqrt{n})$ queries. \
All we know is that the certificate can be checked using $\tO(\sqrt{n})$ queries classically, but this does not imply a quadratic quantum speedup. \
To obtain a quadratic speedup the certificates need some structure that can, for example, be exploited with Grover's algorithm. \
But the certificates in this case are simply $0$-inputs of $\ksum$ of size $\tO(\sqrt{n})$, and it is unclear how to quantumly check if an input is a $0$-input (even with query access to any certificate) any quicker than querying the entire input. \
What we would like is for the certificate to be checkable using $\tO({n}^{1/4})$ queries to the certificate and the input. \
So we construct a new function that has this property by modifying the $\BKK$ function used in \sec{Qvsdeg}.

Let $f$ be $\BKK:\B^{n^2} \to \B$ for some large, but constant $n$. \
We choose $n$ to be large enough that $\Adv(f) \geq n^{1.99}$ and $C(f) \leq n^{1.01}$. \
Such an $n$ exists because asymptotically $C(\BKK) = \tO(n)$ and $Q(\BKK) = \tOmega(n^2)$ (and $Q(\BKK) = \Theta(\Adv(\BKK))$ \cite{LMR+11}). \
Composing this function with itself $d$ times gives us a new function $g:\B^N \to \B$, where $N = n^{2d}$. \
This function has $C(g) \leq C(f)^d \leq n^{1.01d} \leq N^{0.51}$. \
Since the general adversary bound satisfies a perfect composition theorem \cite{LMR+11}, we have
$Q(g) = \Omega(\Adv(g)) = \Omega(\Adv(f)^d) = \Omega(n^{1.99d}) = \Omega(N^{0.99})$.

If we view the function $g$ as a tree of depth $d$ and fanin $n^2$, it has $N$ leaves but only $N^{0.51}$ of these leaves are part of a certificate. \
Consider the subtree of $g$ that consists of all  nodes that are ancestors of these $N^{0.51}$ certificate leaves. \
This subtree has size $O(N^{0.51})$ and the set of values of all nodes in this subtree is a certificate for $g$. \
Furthermore, this certificate can be checked by a quantum algorithm in only $O(N^{0.26})$ queries, since each level of the tree consists of at most $N^{0.51}$ instances of $f$ acting on a constant number of bits. \
Checking a constant-sized $f$ costs only $O(1)$ queries, and searching over all $N^{0.51}$ instances for an invalid certificate costs $O(N^{0.26})$ queries by Grover's algorithm. \ This can be done for each level, resulting in a quantum algorithm with query complexity $\tO(N^{0.26})$.

Thus we have constructed a function $g$ for which $Q(g)$ is close to $N$, but the complexity of quantumly checking a certificate (given query access to it) is close to $N^{1/4}$. \ Plugging this into the cheat sheet construction yields a nearly \th{4} power separation between quantum query complexity and approximate degree.

%%%%%%%%%%%%%%%%%%%%%%%%%%%%%%%%%%%%%%%%%%%%%%%%%%%%%%%%%%%%%%%%%%%%%%%%%%%%%%
\subsection{Implementation}
\label{sec:imp}

Recall the function $\BKK:\B^{n^2} \to \B$ introduced in \thm{BKK}. \ We introduce a function we call $\RBKK$ which consists of
recursively composing $\BKK$ with itself $d$ times; that is, we replace each bit in $\BKK$ with a new copy of $\BKK$, then replace each bit in the new copies with yet more copies of $\BKK$, and so on. \
The resulting function will look like a tree of height $d$ where each vertex has $n^2$ children and will have total
input size $N = n^{2d}$. \
We will choose $d=(4/25)\log n/\log\log n$ to optimize our construction. \
The resulting function $\RBKK$ has the following properties.

\begin{theorem}
\label{thm:RBKK}
    There exists a total function $\RBKK:\B^N \to \B$
    such that given query access to a certificate of size $N^{1/2+o(1)}$, a quantum algorithm can check the validity of the certificate using
    at most $N^{1/4+o(1)}$ queries. \ Furthermore, $Q(\RBKK) = N^{1-o(1)}$.
\end{theorem}

\begin{proof}
Our function $\RBKK$ is defined as the $d$-fold composition of $\BKK:\B^{n^2}\to \B$ with itself. \ This yields a function on $N = n^{2d}$ bits and we set $d=(4/25)\log n/\log\log n$.

With this choice of $d$, we can show more precisely that the function $\RBKK$ has a certificate that can be checked by a quantum algorithm using $O(N^{1/4}L_N[1/2,1])=N^{1/4+o(1)}$ queries, where $L_N[a,c]=\exp((c+o(1))\log^a N\log\log^{1-a} N)$, and
$Q(\RBKK)=\Omega({N}/{L_N[1/2,1]}) = N^{1-o(1)}$.

With this choice of $d$, the parameters are now related as follows: Since $N=n^{2d}$, we have
$\log N=2d\log n=(8/25)\log^2 n/\log\log n$. \ This gives
$\log n=(5/4+o(1))\sqrt{\log N\log\log N}$,
so we get $n%
=L_N[1/2,5/4]$. \
Additionally, since
$(\log n)^d=\exp(d\log\log n)=\exp((4/25)\log n)=n^{4/25}$, it follows that $(\log n)^d=L_N[1/2,1/5]$.

We start with the quantum lower bound for $Q(\RBKK)$. \
By \thm{BKK} and the optimality of the general adversary bound \cite{LMR+11}, we have
\begin{equation}
\Adv(\BKK) = \Omega(Q(\BKK))=\Omega\left(
    \frac{n^2}{\log^5 n}\right)
    =\frac{n^2}{(\log n)^{5+o(1)}}.\end{equation}
By using the fact that $\Adv(f^d) = \Adv(f)^d$ for Boolean functions $f$ \cite{LMR+11}, we get
\begin{equation}
Q(\RBKK)=\Omega(\Adv(\BKK)^d)
=\frac{n^{2d}}{(\log n)^{(5+o(1))d}}
=\frac{N}{L_N[1/2,1/5]^{5+o(1)}}=
\Omega({N}/{L_N[1/2,1]}).
\end{equation}

Now we show the upper bound on quantumly checking a certificate. \
First note that every non-leaf node in the $\RBKK$ tree corresponds to a $\BKK$ instance. \
For each such node, there is therefore a set of $C(\BKK)$ children that constitute a
certificate for that $\BKK$ instance. \ We can therefore certify the $\RBKK$ instance by starting from the top
of the tree, listing out $C(\BKK)$ children that constitute a certificate, then listing out $C(\BKK)$
children for each of those, and so on. \
In each layer $i$ of the tree, we thus have $C(\BKK)^i$ nodes that belong to a certificate.

We require our quantumly checkable certificate to provide, for each non-leaf node that belongs to a certificate, pointers to $C(\BKK)$ of the node's children that constitute a certificate for that $\BKK$ instance, starting with the root of the tree. \
A quantum algorithm can then use Grover search to search for a bad certificate. \
More precisely, the algorithm checks the certificate for the root to see if the $C(\BKK)$ children of the root pointed to and their claimed values do indeed form a certificate. \ It then checks if all the claimed values in the first level are correct assuming the claimed values of nodes in the second level and so on. \ The total number of certificates to check is
\begin{equation}
1+C(\BKK)+C(\BKK)^2 +\dots+C(\BKK)^{d-1}\leq 2C(\BKK)^{d-1},
\end{equation}
where each certificate has size $C(\BKK)$. \ Therefore, the quantum algorithm will make
\begin{equation}
\tO(C(\BKK)^{(d-1)/2}C(\BKK)) =\tO(C(\BKK)^{(d+1)/2})
\end{equation}
queries, where we have logarithmic factors due to the pointers being encoded in binary and error reduction. \
From \thm{BKK}, we have
$C(\BKK)=O(n\log^3 n)=n(\log n)^{3+o(1)}$, so the search takes
\begin{equation}
\tO(n^{(d+1)/2}(\log n)^{(3+o(1))(d+1)/2})
    =\tO(N^{1/4}n^{1/2}(\log n)^{(3/2+o(1))d})
    =\tO(N^{1/4}L_N[1/2,37/40]).\end{equation}
quantum queries, which is is $O(N^{1/4}L_N[1/2,1])$.
\end{proof}

Using this function $\RBKK$ we can now establish \thm{Qvsadeg}. \ This proof is very similar to the separation between quantum query complexity and exact degree in \sec{Qvsdeg}.

Let $f$ be the total function $\RBKK:\B^n \to \B$ with $Q(f) = n^{1-o(1)}$ and $C(f) = n^{1/2+o(1)}$, and more importantly given query access to this certificate, a quantum algorithm can check its validity using $n^{1/4+o(1)}$ queries. \
Let $f_\CS$ denote the cheat sheet version of $f$ created using $10\log n$ copies of $f$ as before. \
From \lem{Qlowerbound} we know that $Q(f_\CS) = \Omega(Q(f))$ and hence $Q(f_\CS) = n^{1-o(1)}$. \ We now show $\adeg(f_\CS) = n^{1/4+o(1)}$.

Let $\ell \in [n^{10}]$ be a potential location of the cheat sheet. \
For any $\ell$, consider the problem of outputting $1$ if and only if $f_\CS(z)=1$ and $\ell$ is the location of the cheat sheet for the input $z$. \
For any fixed $\ell$, this can be solved by a quantum algorithm $\mathcal{Q}_\ell$ that makes $n^{1/4+o(1)}$ queries as shown in \thm{RBKK}, since it can simply check if the certificate stored at cell $\ell$ in the array is valid: it can check all the $10 \log n$ certificates for each of the $10 \log n$ instances of $f$ and then check if the outputs of $f$ evaluate to the location $\ell$.

Since approximate polynomial degree is at most quantum query complexity, we can construct a representing polynomial for $\mathcal{Q}_\ell$ for any location $\ell$. \
This is a polynomial $p_\ell$ of degree $n^{1/4+o(1)}$ such that $p_\ell(z)=1$ if and only if $f_\CS(z)=1$ and $\ell$ is the position of the cheat sheet on input $z$. \
Now we can simply add all the polynomials $p_\ell$ together to obtain a new polynomial $q$ of the same degree, except that we first reduce the error in each polynomial to below $1/n^{10}$ so that the total error is bounded. \
(This can be done, for example, using the amplification polynomial construction of \cite[Lemma 1]{BNRdW07}.) \
Now $q(z) = f_\CS(z)$ as argued in \sec{Qvsdeg}.

%%%%%%%%%%%%%%%%%%%%%%%%%%%%%%%%%%%%%%%%%%%%%%%%%%%%%%%%%%%%%%%%%%%%%%%%%%%%%%
\section{Cheat sheet functions}
\label{sec:cheat}

In this section we define the cheat sheet framework more generally. \ Once the framework is set up, proofs based on the framework are short and conveniently separate out facts about the framework from results about the separation under consideration. \ To demonstrate this, we reprove some known separations in \sec{oldsep}.

%%%%%%%%%%%%%%%%%%%%%%%%%%%%%%%%%%%%%%%%%%%%%%%%%%%%%%%%%%%%%%%%%%%%%%%%%%%%%%
\subsection{Certifying functions and cheat sheets}

Intuitively, a certifying function for a function
$f:D\to\B$, where $D \subseteq [M]^n$, is a function $\phi$
that takes in an input $x \in [M]^n$, a claimed value $y_1$ for
$f(x)$, and a proof $(y_2, y_3, \ldots, y_k)$ that
we indeed have $x \in D$ and $f(x)=y_1$. \
The value of the certifying function
will be 1 only when these conditions hold. \ We would also like
the certifying function to depend nontrivially on the certificate
$y$, i.e., for every $x \in D$ there should be some $y$ that makes
the function output $1$, and some $y$ that makes it
output $0$. \ For convenience of analysis, we enforce that
the certificate $y=0^k$ is invalid for all $x$. \
Any nontrivial certifying function can be
made to have this property by requiring that (say) the second bit of $y$, $y_2$,
is $1$ for all valid certificates.

\begin{definition}[Certifying function]
Let $f:D\to\B$ be a partial function, where $D \subseteq [M]^n$ for some integer $M\geq2$. \
We say a total function $\phi:[M]^n\times\B^k\to\B$, where $k$ is any positive integer,
is a \emph{certifying function} for $f$ if the following conditions are satisfied:
\begin{enumerate}[itemsep=0pt]
\item $\forall x\notin D$, $\phi(x,y)=0$ \hfill (invalid inputs should not have certificates)
\item $\forall x\in[M]^n$, $\forall y\in\B^k$, if $y_1 \neq f(x)$ then $\phi(x,y)=0$ \hfill (certificate asserts $y_1 = f(x)$)
\item $\forall x\in D$, $\exists y\in\B^k$ such that $\phi(x,y)=1$ \hfill (valid inputs should have certificates)
\item $\forall x \in [M]^n$, $\phi(x,0^k)=0$ \hfill (nontriviality condition)
\end{enumerate}
\end{definition}

Typically, a certifying function will be defined so
that its value is only $1$ if $y$ includes
pointers to a certificate in $x$ and if $f$ is a partial function,
we will also want $y$ to include a proof that $x$ satisfies the promise of $f$. \
Thus the query complexity of $\phi$ may be smaller than that of $f$.

%%%%%%%%%%%%%%%%%%%%%%%%%%%%%%%%%%%%%%%%%%%%%%%%%%%%%%%%%%%%%%%%%%%%%%%%%%%%%%

We now define the cheat sheet version of a function $f$ with certifying function $\phi$. \
In the separations in the previous sections we used $10\log n$ copies of the function $f$ to create the address
of the cheat sheet. \
For generality we now use $c = \lceil 10 \log D(f) \rceil$ instead so that the construction only adds logarithmic factors in $D(f)$, as opposed to logarithmic factors in $n$, which may be larger than $D(f)$. \
This does not make a difference in our applications, however.

\begin{definition}[Cheat sheet version of $f$]
Let $f:D\to\B$ be a partial function, where $D \subseteq [M]^n$ for some integer $M\geq2$, and let $c = \lceil 10 \log D(f) \rceil$. Let $\phi:[M]^n\times\B^k\to\B$, for some positive integer $k$, be a certifying function for $f$. \
Then \emph{the cheat sheet version of $f$ with respect to $\phi$}, denoted $f_\CSp$, is a total function
\begin{equation}
f_\CSp:([M]^n)^c\times((\B^k)^c)^{2^c}\to\B
\end{equation}
acting on an input $(X,Y_1,Y_2,\ldots, Y_{2^c})$, where $X \in ([M]^n)^c$ and $Y_i \in (\B^k)^c$. \
Also, let $X \in ([M]^n)^c$ be written as $X=(x^1,x^2,\ldots, x^c)$, where $x^i \in [M]^n$. \
Then we define the value of $f_\CSp(X,Y_1,Y_2,\ldots Y_{2^c})$ to be $1$ if and only if conditions (1) and (2) hold.
\begin{enumerate}[topsep=4pt]
\item[(1)] For all $i\in [c]$, $x^i \in D$.
\end{enumerate}
If condition (1) is satisfied, let $\ell \in [2^c]$ be the positive integer corresponding to the binary string $(f(x^1),f(x^2),\ldots,f(x^c))$ and let $Y_\ell = (y^1, y^2, \ldots y^c)$ where each $y^i \in \B^k$.
\begin{enumerate}[topsep=4pt]
\item[(2)] For all $i \in [c]$, $x^i \in D$ and $\phi(x^i,y^i) = 1$.
\end{enumerate}
\end{definition}

Intuitively, the input to $f_\CSp$ is interpreted as $c$ inputs to $f$,
which we call $(x^1,x^2,\ldots,x^c)$, followed by $2^c$ strings
$(Y_1,Y_2,\ldots Y_{2^c})$ of length $ck$ called cheat sheets. \
The value of the function will be $0$ if any of the $c$ inputs to $f$
do not satisfy the promise of $f$. \
If they do satisfy the promise, let $\ell \in [2^c]$ be the number encoded by
the binary string $f(x^1)f(x^2)\dots f(x^c)$, where $0^c$ and $1^c$ encode the numbers
$1$ and $2^c$ respectively. \
If the \th{\ell} cheat sheet in the input, $Y_\ell$, is written as
$Y_\ell=(y^1, y^2, \ldots, y^c)$, where $y^i\in\B^k$ for all $i$, then the value of
$f_\CSp$ is $1$ if $\phi(x^i,y^i)=1$
for all $i \in [c]$, and $0$ otherwise.

In other words, the value of $f_\CSp$ is $1$ if
the first part of its input consists of $c$ valid
inputs to $f$ that together encode a location of a
cheat sheet in the second part of the input, and
this cheat sheet in turn
contains pointers that certify
the values of the $c$ inputs to $f$. \ The idea is that
the only way to find the cheat sheet is to solve
the $c$ copies of $f$, so that $f_\CSp$ has query
complexity at least that of $f$, but once the cheat sheet
has been found one can  verify the value of
$f_\CSp$ using only the query complexity of $\phi$,
which may be much smaller.

We now characterize the query complexity of $f_\CSp$
in terms of the query complexities of $f$ and $\phi$ in
various models. \ The following result summarizes our results
for the query complexity of cheat sheet functions.

\begin{theorem}[Query complexity of cheat sheet functions]
\label{thm:cs}
Let $f:D\to\B$, where $D \subseteq \B^n$, be a partial function,
and let $\phi$ be a certifying function for $f$. \
Then the following relations hold.
\begin{itemize}[noitemsep,topsep=2pt]
\item $D(f_\CSp)=\tTheta(D(f)+D(\phi))$
\item $R(f_\CSp)=\tTheta(R(f)+R(\phi))$
\item $Q(f_\CSp)=\tTheta(Q(f)+Q(\phi))$
\item $\deg(f_\CSp)=\tTheta(\deg(\phi))$
\item $\adeg(f_\CSp)=\tTheta(\adeg(\phi))$
\item $R_0(f_\CSp)=\tOmega(R_0(f)+R_0(\phi))$ (but the corresponding upper bound relation is false)
\item $Q_E(f_\CSp)=\tO(Q_E(f)+Q_E(\phi))$ (we conjecture that the corresponding lower bound holds)
\end{itemize}
\end{theorem}

The algorithmic measures, $D$, $R$, and $Q$ behave as expected
since \thm{cs} asserts that the straightforward way to compute $f_\CSp$
in these models, which is to first compute all $c$ copies of $f$ and then
check if they point to a valid cheat sheet, is essentially optimal. \
We conjecture that $Q_E$, the quantum analogue of $D$,
should behave similarly to $D$ or $Q$, but we can only prove the upper bound (see \lem{upper}).
In fact, if the lower bound for $Q_E$ can be proved, then we would get a near cubic separation between
$Q(f)$ and $Q_E(f)$.

For the zero-error class $R_0$, while we can show that the lower bound holds,
the upper bound relation is false (see \thm{R0counter} for a counterexample). \
We believe the zero-error class $Q_0$
behaves similarly: we conjecture the analogous lower bound
holds for $Q_0$, but are unable to prove it. \ For $Q_0$, the
corresponding upper bound is false, for reasons similar to those for $R_0$.

Notably, one-sided error measures do not behave well for cheat sheet functions at all. \
In \thm{R1lowercounter}, we demonstrate $f$ and $\phi$ such that
$R_1(f_\CSp)\ll R_1(f)$, showing that the lower bound fails. \
The upper bound need not hold either.

Lastly, note that $\deg$ and $\adeg$ behave fundamentally
differently on cheat sheet functions than the
algorithmic measures. \ The (approximate) degree of the function
$f$ does not even play a role in determining the (approximate) degree
of $f_\CSp$.

\thm{cs} is proved in the next two sections. \
In \app{curious} we present some examples
where the na\"{i}vely expected relations do not hold.

%%%%%%%%%%%%%%%%%%%%%%%%%%%%%%%%%%%%%%%%%%%%%%%%%%%%%%%%%%%%%%%%%%%%%%%%%%%%%%
\subsection{Upper bounds on the complexity of cheat sheet functions}
\label{sec:upper}

We start by showing that the usual algorithmic models
can compute $f_\CSp$ in roughly the number of queries
they require for $f$ and $\phi$.

\begin{lemma}
\label{lem:upper}
Let $f:D\to\B$, where $D \subseteq [M]^n$, be a partial function, and $\phi$ be a certifying function for $f$. \ Then the following upper bounds hold.
\begin{itemize}[noitemsep,topsep=0pt]
\item $D(f_\CSp)=\tO(D(f)+D(\phi))$
\item $R(f_\CSp)=\tO(R(f)+R(\phi))$
\item $Q(f_\CSp)=\tO(Q(f)+Q(\phi))$
\item $Q_E(f_\CSp)=\tO(Q_E(f)+Q_E(\phi))$
\end{itemize}
\end{lemma}

\begin{proof}
The algorithm to evaluate $f_\CSp$ is the following:
First, evaluate $f$ on the $c$ inputs to $f$ contained
in $z$, assuming they satisfy the promise of $f$. \
If one of these inputs is discovered to contradict
the promise of $f$, then output $0$ and halt. \
This takes $cD(f)$ queries for a deterministic algorithm (and $cQ_E(f)$ queries for an exact quantum algorithm),
$O(R(f)c\log c)$ for a bounded-error randomized algorithm, and $O(Q(f)c\log c)$ for a
bounded-error quantum algorithm, where the factor of $O(\log c)$ comes
from reducing the error enough to ensure that all $c$ functions
are computed correctly with high probability. \
From the $c$ outputs we get a binary number $\ell$. \ We can then go to the
appropriate cheat sheet, and evaluate $\phi$ on all the
$(x^i,y^i)$ pairs and output $1$ if they are all $1$. \
This takes $cD(\phi)$ queries for a deterministic algorithm (and $cQ_E(\phi)$ queries for an exact quantum algorithm),
$O(R(\phi)c\log c)$ queries for a randomized algorithm,
and $O(Q(\phi)c\log c)$ queries for a quantum algorithm.
Since $\phi(x,y)=0$ when $x$ does not satisfy the promise
of $f$, this algorithm is correct. \
\end{proof}

Before proving the upper bounds for the degree measures, we prove a technical lemma that shows that $\log(D(f))$ and $\log(C(\phi))$ or $\log(\adeg(\phi))$ are the same up to constants. \ We prove this to show that the log factors that appear in our construction really can be neglected compared to the measures we consider.

\begin{lemma}
\label{lem:phi}
Let $f:D\to\B$ be a partial function, where $D \subseteq [M]^n$,
and let $\phi$ be a certifying function for $f$. \
Then $D(f)\leq C^{(1)}(\phi)^2$ and
$C(f)\leq C^{(1)}(\phi)$. \
As a consequence,
we have $D(f)\leq D(\phi)^2$ and
$C(f)\leq C(\phi)$.
\end{lemma}
\begin{proof}
Note that every $1$-certificate for $\phi(x,y)$
proves that $f(x)=y_1$. \ In particular, such a
$1$-certificate must reveal a certificate for $f$ in $x$. \
This already shows that $C(f)\leq C^{(1)}(\phi)$.

Moreover, every $1$-certificate of $\phi(x,y)$
where $y_1=0$ must conflict with every $1$-certificate
of $\phi(x,y)$ where $y_1=1$ on some bit in $x$. \
This is because otherwise, there would be some $x$
in the promise of $f$ that can be certified to be both
a $0$- and a $1$-input for $f$.

We can then use the following strategy to compute
$f(x)$. \ On input $x\in\Dom(f)$, pick a $1$-certificate
of some $\phi(x,y)$ with $y_1=0$, and query all of its
entries in $x$. \ This reveals at least one bit of every
$1$-certificate of $\phi$ that has $y_1=1$. \
We then find another $1$-certificate of some $\phi(x,y)$
where $y_1=0$,
consistent with the inputs revealed in $x$ so far. \
We once again reveal all of its entries in $x$.
After $C^{(1)}(\phi)$ iterations of this process,
we have either eliminated all $1$-certificates of $\phi$
with $y_1=0$, or else all $1$-certificates of $\phi$ with
$y_1=1$. \ This tells us the value of $f(x)$, because
since $x\in\Dom(f)$, there must be some $y$ such that
$\phi(x,y)=1$, and hence there must be some $1$-certificate
of $\phi$ consistent with $x$. \
The total number of queries used was at most
$C^{(1)}(\phi)^2$.
\end{proof}

Next, we show the upper bounds on the degree measures of
$f_\CSp$ and show that they only depend on the degree of $\phi$ and not
on any measure of the function $f$.

\begin{lemma}
Let $f:D\to\B$ be a partial function, where $D \subseteq \B^n$,
and $\phi$ be a certifying function for $f$. \ Then
$\deg(f_\CSp)=\tO(\deg(\phi))$ and
$\adeg(f_\CSp)=\tO(\adeg(\phi))$.
\end{lemma}

\begin{proof}
We start by showing this for $\deg(f_\CSp)$. \
For each $\ell\in[2^c]$, we construct a polynomial
which is $1$ if $f_\CSp(z)=1$ and the cheat sheet
of $z$ is the \th{\ell} cheat sheet, and $0$ otherwise. \
We can do this by simply multiplying together
the polynomials of $\phi(x^i,y^i)$ for all $i \in [c]$,
where $x^i$ are the inputs to $f$ in $z$ and $y^i$ are
the entries found in the \th{\ell} cheat sheet of $z$.

The resulting polynomial, $p_\ell$, has degree $c\deg(\phi)$,
is $1$ on input $z$ if $z$ is a $1$-input for $f_\CSp$
and $\ell$ is the location of the cheat sheet for input $z$,
and is $0$ on all other inputs. \
To get a polynomial for $f_\CSp$, we simply sum up
the polynomials $p_\ell$ for all $\ell\in[2^c]$. \ The resulting
polynomial is equal to $f_\CSp(z)$ on all inputs $z$,
and has degree $c\deg(\phi)$.

The same trick works for $\adeg$, except in that case
we cannot simply construct the \textsc{And} of $c$
polynomials by multiplying them together:
we must first perform error reduction on the polynomials
and map the output range $[0,1/3]\cup[2/3,1]$
to $[0,1/3c]\cup[1-1/3c,1]$. \
It is possible to do this amplification while increasing
the degree of the polynomials by at most $O(\log c)$ using,
for example, the amplification polynomial construction of \cite[Lemma 1]{BNRdW07}. \
Therefore we have $\adeg(f_\CSp)=O(c\adeg(\phi)\log c)=\tO(\adeg(\phi))$.

Finally, note that $c$ is at most logarithmic in $\adeg(\phi)$, since
 $c=O(\log D(f))=O(\log C^{(1)}(\phi))
 =O(\log C(\phi))=O(\log\adeg(\phi))$,
 where the last equality follows from \cite{BdW02} and
 from the fact that $\phi$ is a total function.
\end{proof}

%%%%%%%%%%%%%%%%%%%%%%%%%%%%%%%%%%%%%%%%%%%%%%%%%%%%%%%%%%%%%%%%%%%%%%%%%%%%%%
\subsection{Lower bounds on the complexity of cheat sheet functions}
\label{sec:lower}

\begin{lemma}
\label{lem:lower}
Let $f:D\to\B$ be a partial function, where $D \subseteq \B^n$,
and let $\phi$ be a certifying function for $f$. \
Then the complexity of $f_\CSp$ is at least
that of $\phi$ with respect to any complexity measure that does
not increase upon learning input bits, which includes
all reasonable complexity models such as $Q$, $D$, $R$, $R_0$, $R_1$, $Q_1$, $Q_0$,
$Q_E$, $C$, $RC$, $bs$, $\deg$, and $\adeg$.
\end{lemma}

\begin{proof}
Consider restricting $f_\CSp$ to a promise that
guarantees that all cheat sheets of the input are
identical. \ An input $z$ to this promise problem evaluates
to $1$ if and only if the $\phi$ inputs from the first
cheat sheet all evaluate to $1$. \ We add the additional
promise that this is true for all but the last input to
$\phi$. \ The value of an input $z$ to this function
then becomes exactly $\phi$ evaluated on a certain
portion of the input. \ The complexity measures
mentioned do not increase when the function
is restricted to a promise. \ Thus the complexity
of $f_\CSp$ under these measures is at least
that of $\phi$.
\end{proof}

\begin{lemma}
Let $f:D\to\B$ be a partial function, where $D \subseteq [M]^n$,
and let $\phi$ be a certifying function for $f$. \
Then $D(f_\CSp)=\Omega(D(f))$.
\end{lemma}

\begin{proof}
We describe an adversary strategy that can be used
to force a deterministic algorithm to make $\Omega(D(f))$
queries. \ Let $x^1,x^2,\dots,x^c$ represent the $c$
inputs to $f$. \ Whenever the algorithm queries a bit
not in $x^1,x^2,\dots,x^c$,
the adversary returns $0$. \
Whenever the algorithm queries a bit in
$x^i$, the adversary uses the adversarial
strategy for $f$ on that input; that is,
it returns answers that are consistent with the promise
of $f$, but that force any deterministic algorithm
to use $D(f)$ queries to discern $f(x^i)$.

Now if a deterministic algorithm
uses fewer than $D(f)$ queries
on an input to $f_\CSp$, then there must be
many cheat sheets in the input
that it did not query at all. \
In addition, such an algorithm cannot determine the value
of $f(x^i)$ for any $i$. \ It follows that
this algorithm could not have discerned the value
of $f_\CSp$ on the given input, and thus
$D(f_\CSp)\geq D(f)$.
\end{proof}

\begin{lemma}
Let $f:D\to\B$ be a partial function,
where $D\subseteq [M]^n$,
and let $\phi$ be a certifying function for $f$. \
Then $R_0(f_\CSp)=\Omega(R_0(f))$.
\end{lemma}

\begin{proof}
Let $A$ be a zero-error randomized algorithm
for $f_\CSp$. \ Then $A$ must always find a certificate
for $f_\CSp$. \ Consider running $A$ on an input $z$
consisting of $c$ $0$-inputs to $f$, together with an
all-zeros array. \ Certifying such an input
requires presenting $0$-certificates for some of the
$c$ inputs to $f$, and presenting at least one bit from
each cheat sheet whose index has not been disproven
by the $0$-certificates of the $f$ inputs.

Now, since there are at least $D(f)^{10}$ cheat sheets,
only a small fraction can be provided in a certificate
of size at most $O(R_0(f))$. \ It follows that
the algorithm $A$ must find certificates to at least
half of the $c$ $0$-inputs to $f$ with high probability.

We can now turn $A$ into a new algorithm $B$
that finds a certificate in a given
$0$-input to $f$ sampled from the hard distribution
$\mathcal{D}^0$ over the $0$-inputs of $f$. \
Given such an input $x$, the algorithm $B$ will generate
$c-1$ additional inputs from $\mathcal{D}^0$,
shuffle them together, and add an all-zero array.
$B$ will then feed this input to $A$. \ We know that
with high probability, $A$ will find a certificate for
at least half the inputs to $f$, so it must find
a certificate for $x$ with probability close to $1/2$
(say, probability at least $1/3$). \
If it does not return a certificate, $B$ can repeat this
process. \ The expected running time of $B$ is at most
$3$ times that of $A$.

Similarly, if we feed $A$ an input consisting of $c$
$1$-inputs to $f$ together with an all-zero array,
we can use a similar argument to construct an algorithm
$C$ for finding $1$-certificates in inputs sampled
from $\mathcal{D}^1$ (the hard distribution over
$1$-inputs for $f$). \ We can then alternate running
steps of $B$ and steps of $C$ to get an algorithm
that finds a certificate in an input sampled from
either $\mathcal{D}^0$ or $\mathcal{D}^1$, which uses
an expected number of queries that is at most $6$
times that of $A$.

This last algorithm evaluates $f$ on inputs sampled from
the hard distribution for $f$, and therefore must have
expected running time at least $R_0(f)$. \ It follows
that the expected running time of $A$ is at least
$R_0(f)/6$.
\end{proof}

The lower bounds for bounded-error randomized query complexity and bounded-error quantum query complexity are exactly as in \lem{Rlowerbound} and \lem{Qlowerbound}.

%%%%%%%%%%%%%%%%%%%%%%%%%%%%%%%%%%%%%%%%%%%%%%%%%%%%%%%%%%%%%%%%%%%%%%%%%%%%%%
\subsection{Proofs of known results using cheat sheets}
\label{sec:oldsep}

\begin{theorem}[\cite{ABB+15}]
\label{thm:R0vsQ}
 There exists a total function $f$ such that  $R_0(f) = \tOmega(Q(f)^{3})$.
\end{theorem}
\begin{proof}
This proof is similar to that of \thm{RvsQ}, except with a different function $g$. \
Let $g:D \to \B$, where $D\in \B^n$, be a partial function that satisfies $Q(g) = 1$ and $R_0(g) = \Omega(n)$, such as the Deutsch--Jozsa problem \cite{DJ92,CEMM98}. \ $R_0(g)=\Omega(n)$ follows from the fact that even certifying that the input is all zeros or all ones requires a certificate of linear size.

Let $h: \B^m \to \B$ be the $\ANDOR$ function on $m$ bits. \
Let $f = g \circ h$ and let $\phi$ be a certifying function that requires the obvious certificate for $f$ (as used in the proof of \thm{RvsQ}). \ Then we have $Q(\phi) = O(n+\sqrt{n}m^{1/4})$ and $Q(f) = Q(g \circ h) = O(\sqrt{m})$. \
From \thm{Rcomposition} we have that $R_0(f) = R_0(g\circ h) = \tOmega(nm)$. \
Using \thm{cs}, we have $Q(f_\CSp) = \tO(Q(f) + Q(\phi)) = \tO(\sqrt{m}+n+\sqrt{n}m^{1/4})$ and $R(f_\CSp) = \tOmega(R_0(f)) = \tOmega(nm)$. \
Plugging in $m=n^2$, we get $Q(f_\CSp) = \tO(n)$ and $R(f_\CSp) = \tOmega(n^{3})$.
\end{proof}

\begin{theorem}[\cite{ABB+15}] 
\label{thm:RvsQE}
There exists a total function $f$ such that   $R(f) = \tOmega(Q_E(f)^{3/2})$  .
\end{theorem}

\begin{proof}
Let $g:D \to \B$, where  $D = \{x\in \B^n: |x| = 0 \textrm{ or } |x| = 1\}$, be the partial function that evaluates to $0$ if the input is the all-zeros string and evaluates to $1$ if the input string has Hamming weight $1$. \
This function satisfies $R(g) = \Omega(n)$, since the block sensitivity of the $0$-input is linear and $Q_E(g) = O(\sqrt{n})$, since Grover's algorithm can be made exact in this case \cite{BHMT02}.

Let $h:\B^m \to \B$ be the $\AND$ function on $m$ bits, which satisfies $R(h) = \Omega(m)$ and $Q_E(h) = O(m)$. \
Let $f = g \circ h$, and let $\phi$ be a certifying function for $f$. \
The certificate complexity of this function is $O(n)$ in the no case (by showing one $0$-input to each of the $n$ $\AND$ gates), and is $O(m)$ in the yes case (by showing all the $1$ inputs to an $\AND$ gates). \
We choose the certifying function $\phi$ that requires this certificate for $f$ and hence $Q_E(\phi) \leq D(\phi) = \tO(n+m)$.

Now by composing quantum algorithms for $g$ and $h$ we have $Q_E(f) = O(m\sqrt{n})$, and by \thm{Rcomposition} we have $R(f) = \Omega(mn)$. \
Hence using \thm{cs}, we have $Q_E(f_\CSp) = \tO(Q_E(f)+Q_E(\phi)) = \tO(n + m + m\sqrt{n})$, and $R(f_\CSp) = \tOmega(mn)$. \ Choosing $m=\sqrt{n}$ gives $Q_E(f_\CSp) = \tO(n)$ and $R(f_\CSp) = \tOmega(n^{3/2})$.
\end{proof}

%%%%%%%%%%%%%%%%%%%%%%%%%%%%%%%%%%%%%%%%%%%%%%%%%%%%%%%%%%%%%%%%%%%%%%%%%%%%%%
\section*{Acknowledgements}
We thank Ansis Rosmanis for comments on \app{ksum} and Mark Bun for discussions on approximate degree lower bounds.
We thank Hardik Bansal for spotting an error in an earlier version of the proof of \thm{Rcomposition}. 
We also thank  Iordanis Kerenidis, Fr\'ed\'eric Magniez, and Ashwin Nayak for catching a bug in an earlier version of the proof of \lem{Rlowerbound}.

This work was partially funded by the ARO grant Contract Number W911NF-12-1-0486, as well as by an Alan T.\ Waterman Award from the National Science Foundation,
under grant no.\ 1249349. This preprint is MIT-CTP \#4730.

%%%%%%%%%%%%%%%%%%%%%%%%%%%%%%%%%%%%%%%%%%%%%%%%%%%%%%%%%%%%%%%%%%%%%%%%%%%%%%
\bibliographystyle{alpha}
\phantomsection\addcontentsline{toc}{section}{References} %Adds References to TOC
\bibliography{cheat}

\appendix

%%%%%%%%%%%%%%%%%%%%%%%%%%%%%%%%%%%%%%%%%%%%%%%%%%%%%%%%%%%%%%%%%%%%%%%%%%%%%%
\section{Quantum query complexity of \texorpdfstring{$\ksum$}{k-sum}}
\label{app:ksum}

Belovs and \v{S}palek \cite{BS13} proved a lower bound
of $\Omega(n^{k/(k+1)})$  on the quantum query complexity of $\ksum$
for constant $k$,
as long as the alphabet is large enough. \ However,
in their result, Belovs and \v{S}palek
do not keep track of the dependence on $k$. \
We trace their proof to determine
this dependence, ultimately proving the following theorem.
\begin{theorem}
\label{thm:ksum}
    For the function $\ksum:[q]^n \to \B$, if $|q| \geq 2{\binom{n}{k}}$, we have
    \begin{equation}Q(\ksum)=\Omega(n^{k/(k+1)}/\sqrt{k}),\end{equation}
    where the constant in the $\Omega$ is independent of $k$.
\end{theorem}

\begin{proof}
We proceed by tracing Belovs and \v{S}palek's
proof, which can be found in section $3$
of \cite{BS13}, to extract an unsimplified lower bound. \
We then tune some parameters to determine the best
dependence on $k$.

The proof in \cite{BS13}
proceeds by the generalized adversary method. \
Belovs and \v{S}palek construct an adversary matrix
$\Gamma$, which depends on parameters $\alpha_m$ for
$m=0,1,\dots,n-k$. \
This construction can be found in Section $3.1$ of their
paper. \
They then show a lower bound on $\|\Gamma\|$ and an upper
bound on $\|\Gamma\circ\Delta_1\|$
in terms of the parameters $\alpha_m$
(the upper bounds on $\|\Gamma\circ\Delta_i\|$
for $i\neq 1$ follow by symmetry). \ Finally, they
pick the values for the parameters that give the best
lower bound.

A lower bound on $\|\Gamma\|$ without $\Omega$ notation,
which therefore does not ignore factors of $k$, can be
found in Section $3.6$ of their paper. \ In that section, they
show
\begin{equation}
\|\Gamma\|\geq \alpha_0\sqrt{{\binom{n}{k}}
\left(1-\frac{1}{q}{\binom{n}{k}}\right)},
\end{equation}
where $q$ is the alphabet size. \
When $q\geq \frac{4}{3}{\binom{n}{k}}$, this gives
\begin{equation}
\|\Gamma\|\geq\frac{\alpha_0}{2}\sqrt{{\binom{n}{k}}}.
\end{equation}

For the upper bound on $\|\Gamma\circ\Delta_1\|$,
it turns out to be more
convenient to switch to a different matrix, $\tilde{\Gamma}$,
which satisfies
$\|\Gamma\circ\Delta_1\|\leq\|\tilde{\Gamma}\circ\Delta_1\|$. \
This is explained in the beginning of Section $3$ of their paper and in
Section $3.1$. \ To upper bound
$\|\tilde{\Gamma}\circ\Delta_1\|$, in Section $3.3$
the authors switch to yet another matrix, $\tilde{\Gamma}_1$,
with the property that
$\tilde{\Gamma}_1\circ\Delta_1=\tilde{\Gamma}\circ\Delta_1$
and that
$\|\tilde{\Gamma}_1\circ\Delta_1\|\leq2\|\tilde{\Gamma}_1\|$. \
This gives
$\|\Gamma\circ\Delta_1\|\leq2\|\tilde{\Gamma}_1\|$,
so it suffices to upper bound $\|\tilde{\Gamma}_1\|$.

In Section $3.4$, the authors express
$\|\tilde{\Gamma}_1\|^2$ as the norm of a sum of matrices,
with the sum ranging over a set $\mathcal{S}$. \
They split $\mathcal{S}=\mathcal{S}_1\cup\mathcal{S}_2$;
by the triangle inequality, it then suffices to separately
upper bound the norm of the sum over $\mathcal{S}_1$
and the norm of the sum over $\mathcal{S}_2$. \
The former is upper bounded by
$\max_m\alpha_m^2{m+k-1\choose k-1}$,
and the latter by
$k{n-1\choose k}\max_m(\alpha_m-\alpha_{m+1})^2$. \
We therefore conclude that
\begin{equation}
\|\Gamma\circ\Delta_1\|^2\leq
4\max_m\alpha_m^2{m+k-1\choose k-1} +
4k{n-1\choose k}\max_m(\alpha_m-\alpha_{m+1})^2.
\end{equation}

It remains to pick values for $\alpha_m$ to optimize
the resulting adversary bound. \ There are no
restrictions on the $\alpha_m$ parameters. \
To maximize the adversary bound,
we want $\alpha_0$ to be large, consecutive $\alpha$'s
to be close (to minimize $\max_m(\alpha_m-\alpha_{m+1})^2$),
and $\alpha_m$ for large $m$ to be small
(to minimize $\max_m\alpha_m^2{m+k-1\choose k-1}$). \
We choose to make the $\alpha_m$ parameters
decrease to $0$ in an arithmetic sequence,
with $\alpha_m-\alpha_{m+1}=c>0$
for all $m\leq\alpha_0/c$, and $\alpha_m=0$ for all
$m\geq\alpha_0/c$. \ Let $\beta=\alpha_0/c$. \ Then
$\alpha_m=\alpha_0-cm$ for $m\leq\beta$.

We can write the final adversary bound as follows:
\begin{equation}\Adv(\ksum)^2=\Omega\left(
    \frac{\alpha_0^2{\binom{n}{k}}}
    {\max_{m\leq\beta}(\alpha_0-cm)^2{m+k-1\choose k-1}+
    k{n-1\choose k}c^2}\right)\end{equation}
\begin{equation}=\Omega\left(\min\left\{
    \frac{{\binom{n}{k}}}
    {\max_{m\leq\beta}(1-m/\beta)^2{m+k-1\choose k-1}},
    \frac{{\binom{n}{k}}}{k{n-1\choose k}}\beta^2
    \right\}\right)\end{equation}
The second term in the minimization simplifies to
$\frac{n}{k(n-k)}\beta^2$, which is $\Omega(\beta^2/k)$.
To deal with the first term,
we use the well-known inequalities
$(n/k)^k\leq{\binom{n}{k}}\leq(en/k)^k$,
which hold for all $n$ and $k$. \ This gives
\begin{equation}\Adv(\ksum)^2=\Omega\left(\min\left\{\frac{n^k}
    {(k-1)\max_{m\leq\beta}(1-m/\beta)^2(e(m+k-1))^{k-1}}
    \left(\frac{k-1}{k}\right)^k,
    \frac{\beta^2}{k}\right\}\right)\end{equation}
\begin{equation}=\Omega\left(\min\left\{\frac{n^k}
    {k\max_{m\leq\beta}(1-m/\beta)^2(e(m+k))^{k-1}},
    \frac{\beta^2}{k}\right\}\right).\end{equation}
We can solve the maximization in the denominator using
calculus. \ The unique maximum occurs at
$m=\beta-\left(\frac{2}{k+1}\right)(\beta-k)$. \
Substituting this in and simplifying, we get
\begin{equation}\Adv(\ksum)^2=\Omega\left(\min\left\{
    \frac{k\beta}{(1-k/\beta)^2}
    \left(\frac{n}{e(\beta+3)}\right)^k,
    \frac{\beta^2}{k}\right\}\right)
    =\Omega\left(\min\left\{
    k\beta\left(\frac{n}{3\beta}\right)^k,
    \frac{\beta^2}{k}\right\}\right),\end{equation}
where for the last equality we assumed $\beta\geq 2k$ and
$\beta\geq 3e/(3-e)$.

Finally, we set $\beta=(1/3)n^{k/(k+1)}$. \ This gives
\begin{equation}\Adv(\ksum)^2=\Omega\left(\min\left\{kn^{k/(k+1)}
    \left(\frac{n}{n^{k/(k+1)}}\right)^k,
    \frac{n^{2k/(k+1)}}{k}\right\}\right)
    =\Omega(n^{2k/(k+1)}/k),\end{equation}
so
\begin{equation}Q(\ksum)=\Omega(\Adv(\ksum))
    =\Omega\left(n^{k/(k+1)}\middle/\sqrt{k}\right).\end{equation}

Note that we assumed $\beta\geq 2k$. \
Since $\beta=n^{k/(k+1)}/3$,
we need $n^{k/(k+1)}\geq6k$, or $n\geq(6k)^{1+1/k}$. \
Since
$(6k)^{1/k}=\exp(\ln(6k)/k)=1+\Theta(\log(k)/k)$,
it suffices to have $n\geq 6k+\omega(\log k)$. \
In particular, this bound works as long as we have
$k\leq n/10$. \ When $k\geq n/10$,
we can directly prove a lower
bound of $\Omega(\sqrt{n})$ by a reduction from Grover
search; this means there are no restrictions on the size
of $n$ and $k$ in this result.
\end{proof}

%%%%%%%%%%%%%%%%%%%%%%%%%%%%%%%%%%%%%%%%%%%%%%%%%%%%%%%%%%%%%%%%%%%%%%%%%%%%%%
\section{Measures that behave curiously with cheat sheets}
\label{sec:counter}
\label{app:curious}

In this appendix we show that $R_1$ and $Q_1$ can behave
strangely on cheat sheet functions, potentially
decreasing from $R_1(f)$ to $R(f)$ and from $Q_1(f)$ to
$Q(f)$.

\begin{theorem}\label{thm:R1}
Let $f:D\to\B$ be a partial function, where $D\subseteq [M]^n$,
and let $\phi$ be a certifying function for $f$. \
Then
\begin{equation}
R_1(f_\CSp)=\tO(R(f)+R_0(\phi))
\qquad \text{and} \qquad
Q_1(f_\CSp)=\tO(Q(f)+Q_0(\phi)).
\end{equation}
\end{theorem}

\begin{proof}
We show a randomized (and quantum) algorithm
for $f_\CSp$ that uses the required number of queries
and finds a $1$-certificate with constant probability. \
Such an algorithm could be made to have one-sided
error by outputting $1$  only if it finds a
$1$-certificate.

The algorithm works as follows. \ First, use the randomized
(resp.\ quantum) algorithm on the $c$ inputs to $f$
(with some amplification), to determine the number $\ell$
of the correct cheat sheet with high probability
(assuming the promises of $f$ all hold). \
Next, go to the \th{\ell} cheat sheet, and use a zero-error
algorithm to evaluate
$\phi(x_i,y_i)$ for $i=1,2,\dots, c$, where $x_i$
are the inputs to $f$ and $y_i$ are the cheat sheet
strings. \ This finds certificates for each input to $\phi$.

Now, if the value of the function is $1$ on the given
input, then with constant probability, all the $c$
certificates found should be $1$-certificates for $\phi$. \
These $1$-certificates certify the value of the inputs
to $f$. \ Therefore, taken together, they constitute
a valid $1$-certificate for $f_\CSp$. \
It follows that this algorithm uses
$\tO(R(f)+R_0(\phi))$
randomized queries and finds a $1$-certificate
with constant probability. \ The result for quantum
algorithms follows similarly.
\end{proof}

We now show that it is possible for $R_1(f_\CSp)$ to be much smaller than $R_1(f)$, unlike some of the other measures studied. \ 
Intuitively, this is because $1$-inputs of $f_\CSp$ contains a cheat sheet with a certificate and hence even if the algorithm is not sure of its answer before finding the cheat sheet, the cheat sheet may convince the algorithm of its answer.

\begin{theorem}
\label{thm:R1lowercounter}
There is a total function $f:\B^n\to\B$ and a certifying
function $\phi$ for $f$ such that
$R_1(f_\CSp)=\tO(\sqrt{R_1(f)})$.
\end{theorem}

\begin{proof}
We can use \thm{R1} to construct an
explicit function $f$ and a certifying function $\phi$
for $f$ such that $R_1(f_\CSp)$ is smaller
than $R_1(f)$. \ The construction is as follows. \
In \cite{ABB+15}, a Boolean
function was constructed that has
$R_0=\tOmega(n^2)$ and $R_1=\tO(n)$. \
This function has certificate complexity roughly
equal to $n$. \
Now, by taking the \textsc{Xor} of this function
with an additional bit, we get a function $f$
for which $R_1$ is roughly $n^2$ but $R$ and $C$ are still
roughly $n$.

We define $\phi$ to be a certifying
function that simply checks if $y$ provides pointers
to a certificate for $f$. \ Then $R_0(\phi)=\tO(n)$.
It follows that $R_1(f_\CSp)=\tO(n)$,
while $R_1(f)=\tOmega(n^2)$.
\end{proof}

Lastly we show that even zero-error randomized query complexity behaves curiously with cheat sheets. \ 
We might have expected that the obvious upper bound $R_0(f_\CSp) = \tO(R_0(f) + R_0(\phi))$ holds, but in fact it does not. \ 
The reason is subtle and relates to the behavior of the zero-error algorithm on inputs outside of the domain $D$. \ 
For a partial function $f:D\to \B$, a zero-error algorithm's behavior is not constrained on inputs outside the domain $D$. \
The algorithm could, for example, run forever on such inputs. \ 
The obvious algorithm, which simply runs the zero-error algorithm for $f$ on all $c$ inputs to $f$ now fails to output any answer on $0$-inputs to $f_\CSp$ in which the inputs to $f$ do not lie in $D$. \ 
We exploit this observation to prove the following counterexample. \ 

\begin{theorem}
\label{thm:R0counter}
There is a partial function $f:D\to\B$, where $D\subseteq \B^n$, and a certifying
function $\phi$ for $f$ such that
$R_0(f_\CSp)=\tOmega(R_0(f)^{3/2}+R_0(\phi)^{3/2})$.
\end{theorem}

\begin{proof}
Let $g:\B^{4m}\to\B$ be the partial function
defined as follows. \ On input $(x,y)$ with $x,y\in\B^{2m}$,
define $g(x,y)=0$ if the Hamming weight of $x$ is $m$
and the Hamming weight of $y$ is $0$, and define
$g(x,y)=1$ if the Hamming weight of $y$ is $m$ and the
Hamming weight of $x$ is $0$. \ The promise of $g$ is that
one of these two conditions holds for every input
$(x,y)$.

It is easy to see that $R_0(g)=O(1)$. \ Let $f$ be the
composition of $g$ with an $m$ by $m$ $\ANDOR$,
and let $n=4m^3$. \ Then $R_0(f)=O(m^2)=O(n^{2/3})$. \
Let $\phi$ be a certifying function for $f$ that
takes an input to $f$ and an additional string,
and asserts that the latter provides
pointers to a certificate for each $\ANDOR$
instance in the input to $f$. \ This assertion can
be verified (or refuted) in $\tO(m^2)$
deterministic queries, so
$R_0(\phi)\leq D(\phi)=\tO(m^2)$.

We now show that $R_0(f_\CSp)=\tOmega(m^3)$. \
Consider an input consisting of $c$ inputs
to $f$, each of which has all its $\ANDOR$ instances
evaluate to $0$. \ Note that these inputs to $f$ do not
satisfy the promise of $f$. \ We attach a blank cheat sheet array
after these $c$ inputs. \ By definition of $f_\CSp$, the
resulting input is therefore a $0$-input to $f_\CSp$. \
However, any $0$-certificate
of it (of reasonable size) must
``partially certify'' at least half
of the $c$ inputs to $f$; that is, for at least half
the inputs to $f$, it must either prove the input
is not a $0$-input or prove it is not a $1$-input. \
The only way to do this is to display $(c/2)m$
$0$-certificates for $\ANDOR$ instances.

This means we have an algorithm that takes
in $4cm$ zero-inputs to $\ANDOR$, and finds a certificate
for at least $cm/8$ of them using $R_0(f_\CSp)$
expected queries. \ It follows that given one
input from the hard distribution over $0$-inputs to $f$,
we can find a certificate for it by
generating $4cm-1$ additional inputs from this distribution,
mixing them together, and running the $R_0(f_\CSp)$
algorithm. \ This finds a certificate with probability at
least $1/8$, and uses $R_0(f_\CSp)/cm$ expected
queries. \ By repeating the algorithm if necessary,
we find a certificate using $8R_0(f_\CSp)/cm$
expected queries.

We construct a function $f^\prime$ is the same as $f$
except with \textsc{Not-And-Or} instead of $\ANDOR$. \
Repeating the argument provides a randomized algorithm
that finds a $1$-certificate for $\ANDOR$ using
$8R_0(f^\prime_{CS(\phi^\prime)}/cm$ expected queries. \
By running both algorithms in parallel, we get
\begin{equation}
R_0(\ANDOR)\leq 8R_0(f_\CSp)/cm +8R_0(f^\prime_{\CS(\phi^\prime)})/cm.
\end{equation}
Since $R_0(\ANDOR)=\Omega(m^2)$, it follows that either
$R_0(f_\CSp)=\Omega(m^3)$
or
$R_0(f^\prime_{\CS(\phi^\prime)})=\Omega(m^3)$. \
The desired result follows.
\end{proof}

\end{document}